\documentclass[aps,prl,twocolumn,superscriptaddress,longbibliography,reprint]{revtex4-2}

\usepackage[T1]{fontenc}
\usepackage{amsmath}
\usepackage{amssymb}
\usepackage{amsfonts}
\usepackage{amsthm}
\usepackage{mathtools}
\usepackage{bm}
\usepackage{bbm}
\usepackage{graphicx}
\usepackage[dvipsnames]{xcolor}

\usepackage{times}

\definecolor{myrefcolor}{rgb}{0.067,0.5,0.5}

\newtheorem{theorem}{Theorem}
\newtheorem{corollary}{Corollary}
\newtheorem{proposition}{Proposition}
\newtheorem{lemma}{Lemma}

\newtheorem{conjecture}{Conjecture}

\newcommand{\CC}{\mathbb C}
\newcommand{\E}{\mathbb E}
\newcommand{\id}{\mathbbm 1}
\newcommand{\Schur}{\mathbb S}
\newcommand{\A}{\mathcal A}
\newcommand{\CN}{\mathcal C}
\newcommand{\GB}{G_{\mathrm B}}
\newcommand{\tr}{\operatorname{Tr}}
\newcommand{\ket}[1]{\lvert #1\rangle}
\newcommand{\bra}[1]{\langle #1\rvert}
\newcommand{\braket}[2]{\langle #1\,\vert\,#2\rangle}
\newcommand{\matel}[3]{\langle #1\vert\,#2\,\vert #3\rangle}
\newcommand{\wedgepow}[1]{\textstyle\bigwedge^{#1}}

\newcommand{\Gr}{\mathrm{Gr}}
\newcommand{\Flag}{\mathrm{Fl}}
\newcommand{\dTV}{d_{\mathrm{TV}}}
\newcommand{\Tw}{\mathcal T}

\DeclareMathOperator{\Var}{Var}
\DeclareMathOperator{\SU}{SU}
\DeclareMathOperator{\U}{U}

\usepackage[
  breaklinks=true,
  colorlinks=true,
  linkcolor=myrefcolor,
  citecolor=myrefcolor,
  urlcolor=myrefcolor
]{hyperref}

\hypersetup{
  pdftitle={The magic cost of fermionic randomness},
  pdfauthor={Chenfeng Cao, Yifan Tang, and Jens Eisert}
}

\begin{document}

\title{The magic cost of fermionic randomness}
\author{Chenfeng Cao}
\email{chenfeng.cao@fu-berlin.de}
\affiliation{HK Institute of Quantum Science \& Technology, The University of Hong Kong, Hong Kong, China}
\affiliation{Dahlem Center for Complex Quantum Systems, Freie Universit\"{a}t Berlin, 14195 Berlin, Germany}

\author{Yifan Tang}
\affiliation{Dahlem Center for Complex Quantum Systems, Freie Universit\"{a}t Berlin, 14195 Berlin, Germany}

\author{Jens Eisert}
\email{jense@zedat.fu-berlin.de}
\affiliation{Dahlem Center for Complex Quantum Systems, Freie Universit\"{a}t Berlin, 14195 Berlin, Germany}
\affiliation{Helmholtz-Zentrum Berlin f\"{u}r Materialien und Energie, 14109 Berlin, Germany}

\begin{abstract}
Sampling-based quantum advantage raises a basic resource question: how much non-Gaussianity is needed to reproduce the random-matrix statistics often associated with quantum chaos? Passive fermionic linear optics preserves Gaussian structure, and even Haar-random mode mixing fails to anticoncentrate for Slater inputs. We show that overcoming this obstruction requires an extensive input resource. For any pure half-filled input state, the one-body defect $\nu$, defined as the summed occupation variance of the natural orbitals, must scale linearly with system size. For products of bounded-size blocks on $4N$ modes, this condition is also sufficient, with normalized collision moment asymptotic to $2N/\nu$, compared with the Porter-Thomas value 2. For four-mode products, the full Porter-Thomas moment hierarchy is recovered precisely when the defect density tends to its maximum. Assuming average-case hardness, approximate classical sampling remains hard near maximal density, with an error tolerance that increases with the density. Two-copy number fluctuations measure $\nu$ directly, and the same quantity lower-bounds the combined cost of non-Gaussian resource modes and parity-preserving local gates in deterministic exact preparation. Fermionic randomness therefore carries an extensive magic cost, with a quantitative link between sampling statistics, preparation resources, and an experimentally accessible witness.
\end{abstract}

\maketitle

Anticoncentration is a key statistical ingredient in standard hardness arguments for approximate quantum sampling~\cite{Hangleiter2023Computational,Aaronson2011The,Bremner2016Average,Bouland2019On}. It requires the probability of a fixed outcome, rescaled by the number of allowed outcomes, to exceed a positive threshold with probability bounded away from zero as the system grows~\cite{Hangleiter2018Anticoncentration}. Haar-random complex states satisfy the stronger benchmark of asymptotically exponential rescaled probabilities, with Porter-Thomas moments $\ell!$. These statistics resemble the behavior one expects in suitable regimes of quantum-chaotic dynamics and provide a precise benchmark for the emergence of randomness. How much non-Gaussianity must an input supply for random passive fermionic dynamics to exhibit either behavior? We address this resource question while keeping the evolution entirely Gaussian and placing all non-Gaussianity in the input.

Passive \emph{fermionic linear optics} (FLO) shows that an exponentially large Hilbert space is not enough. It acts on the full fixed-number sector, but it only rotates the modes: every Slater determinant is mapped to another Slater determinant, and the resulting output distribution can be sampled classically~\cite{Valiant2002Quantum,Terhal2002Classical}. For a Slater-determinant input at half filling of $4N$ modes and a Haar-random mode rotation, the normalized second moment, the collision moment used below, grows as $\Theta(N)$ instead of approaching $2$. Relative to Porter--Thomas statistics, probability remains concentrated by a growing factor. The obstruction is therefore the Gaussian structure preserved by the dynamics. The fermion-sampling scheme of Ref.\ \cite{Oszmaniec2022Fermion} has shown that a particular four-mode non-Gaussian input can restore anticoncentration and support sampling-hardness arguments, but has left open how densely that resource must be supplied.

Classical simulability alone does not explain this obstruction. Clifford circuits are classically simulable, yet the multiqubit Clifford group is a unitary
$3$-design~\cite{Webb2016Clifford,Zhu2017Multiqubit} and therefore anticoncentrates. Polylogarithmically many
non-Clifford gates can already restore Haar overlap statistics in doped Clifford
circuits~\cite{Magni2025Anticoncentration,Magni2025QuantumComplexity,
Leone2026NonClifford}. Non-Gaussianity is the magic
resource for matchgate computation~\cite{Hebenstreit2019All} and is quantified by fermionic resource theories~\cite{Sierant2026Fermionic}. 
Deciding stabilizer-polytope membership for general mixed qubit states is computationally hard under reasonable assumptions~\cite{HardnessMagic}, motivating accessible resource bounds in structured settings.
Can a vanishing density of it likewise suffice? Other fermionic architectures distribute non-Gaussianity throughout the evolution through gates or adaptive
measurements~\cite{Paviglianiti2026Emergence,Cao2026Measurement}. Recent work bounds the non-Gaussian gate count needed for approximate 2-designs~\cite{Trigueros2026Unitary,Leone2026FermionicEntropy}, while pseudorandom fermionic states require an extensive number of local non-Gaussian gates~\cite{Haug2026Practical}. In our setting the
evolution remains passive and the entire resource lies in the input.

We show that a vanishing density of input magic cannot suffice. Even anticoncentration, a much weaker requirement than Porter--Thomas statistics, requires an extensive resource for every pure half-filled input. With one-body correlation matrix $\Gamma$, define the defect $\nu=\tr[\Gamma(\id-\Gamma)]$. It sums the natural-orbital occupation variances and measures the one-body activity available to passive mixing. We prove that every anticoncentrating pure input has $\nu=\Omega(N)$. For products of bounded-size blocks the converse holds: their normalized collision moment obeys $\CN_N=(2N/\nu)[1+o(1)]$ as $\nu\to\infty$, so they anticoncentrate precisely when $\nu=\Theta(N)$, and for four-mode products the full hierarchy of fixed-order Porter--Thomas moments is recovered if and only if $\nu/N\to1$. Beyond products, defect alone is not sufficient. A two-copy beam splitter makes $\nu$ directly visible as half the measured number-difference variance, while the weight of the doubled input in the fully Pauli-blocked (replica-singlet) sector fixes the collision moment exactly. The same defect lower-bounds exact preparation cost, so every anticoncentrating pure input requires an extensive combined supply of non-Gaussian resource modes and parity-preserving gates of bounded support.

\smallskip
\paragraph*{Setup.}
A Haar-random mode unitary $u\in\U(4N)$ induces the passive-FLO circuit $V=\wedgepow{2N}u$ on $2N$ fermions, followed by a measurement of all mode occupations [Fig.~\ref{fig:witness}(a)]. The $D_N=\binom{4N}{2N}$ outcomes $x$ are the half-filled occupation patterns, and an input state $\Psi$ yields the Born probabilities
$p_V(x)=|\matel{x}{V}{\Psi}|^2$. For a fixed outcome, we write $X_x:=D_N\,p_V(x)$ for the rescaled outcome probability, viewed as a random variable over the circuit ensemble. Mode-permutation invariance makes $\E_Vp_V(x)$ independent of $x$, and normalization gives $\E_VX_x=1$. The normalized collision moment is
\begin{equation}
  \CN_N:=\E_V[X_x^2]
  =D_N\,\E_V\!\left[\sum_y p_V(y)^2\right],
  \label{eq:collisiondef}
\end{equation}
where the equality follows from the same symmetry.

With $c_a$ the annihilation operator for mode $a$, set
$\Gamma_{a,b}=\langle c_a^\dagger c_b\rangle$ and define the one-body defect $\nu:=\tr[\Gamma(\id-\Gamma)]$, the trace idempotency defect of $\Gamma$. The eigenvalues of $\Gamma$ are the natural-orbital occupations~\cite{Lowdin1955Quantum}, so the trace formula sums their occupation variances and measures their departure from empty or full. Passive FLO conjugates $\Gamma$ and leaves $\nu$ invariant. For a pure
fixed-number state, $\nu=0$ exactly when $\Gamma$ is a projector, equivalently when the state is a Slater determinant. At half filling, $0\le\nu\le N$, with $\nu=N$ exactly when every natural orbital is half occupied. The defect density is
\begin{equation}
  \gamma:=\frac{\nu}{N}\in[0,1].
  \label{eq:defectdensity}
\end{equation}

The defect also counts active particles. The mean number of particles occupying the $2N$ least occupied natural orbitals, equivalently the mean number of holes in the $2N$ most occupied ones, lies between $\nu/2$ and $2\nu$ (Supplemental Material, SM). Thus $\nu$ measures, within a factor of two, the particle--hole content above the
natural-orbital Fermi sea. It also recovers existing quantifiers: for a fixed-number state, $4\nu$ is the first fermionic antiflatness of Ref.~\cite{Sierant2026Fermionic} and, up 
to normalization, the quadratic one-body-entanglement monotone of Ref.~\cite{Gigena2020OneBody}. At half filling, $\nu/N=1-(4N)^{-1}\tr[(\id-2\Gamma)^2]$ is the one-body-purity deficit in the convention of Ref.~\cite{Swietek2026OneBody}.

For the explicit formulas below, we specialize to products of four-mode blocks. Group the $4N$ modes into $N$ blocks of four [Fig.~\ref{fig:witness}(a)]. By the Youla normal
form~\cite{Youla1961Normal}, equivalently the two-fermion Slater decomposition~\cite{Schliemann2001Quantum}, every pure two-fermion state vector on four modes can be brought by a block-local passive-FLO unitary to the form
\begin{equation}
\ket{\psi_{\lambda_i}}
=
\sqrt{1-\lambda_i}\,\ket{0,0,1,1}
+
\sqrt{\lambda_i}\,\ket{1,1,0,0},
\label{eq:blocknormal}
\end{equation}
where $0\le\lambda_i\le1/2$ after choosing a phase convention and ordering the two Slater pairs (SM). Since $u$ is Haar-random, these local rotations can be absorbed into it, and we use the product representative $\ket{\Psi_{\bm\lambda}}
=\bigotimes_{i=1}^{N}\ket{\psi_{\lambda_i}}$ throughout. Each block contributes $4\lambda_i(1-\lambda_i)\in[0,1]$ to the defect, so $\nu=4\sum_{i=1}^{N}\lambda_i(1-\lambda_i)$. Within this family, $\nu=N$ exactly when $\lambda_i=1/2$ for every block. The uniform family has $\lambda_i=\lambda$ for all blocks.

\begin{figure}[t]
\centering
\includegraphics[width=0.93\columnwidth]{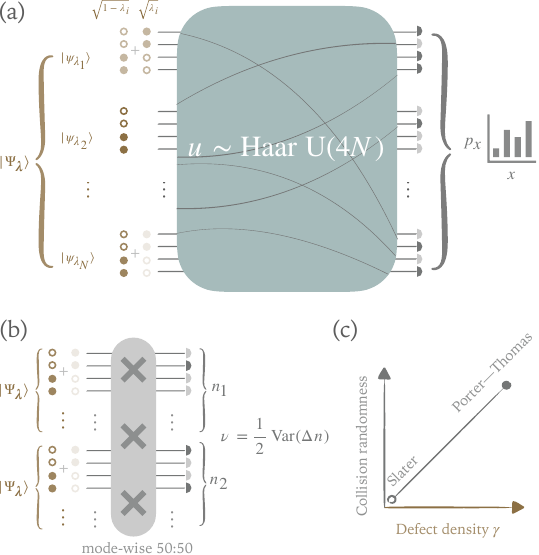}
\caption{\textbf{Fermion sampling and its two-copy witness.}
(a)~A Haar-random mode unitary $u\in\U(4N)$ induces the passive-FLO circuit $V=\wedgepow{2N}u$, which acts on $N$ four-mode blocks $\ket{\psi_{\lambda_i}}$ before an occupation measurement with outcome $x$.
(b)~Two identical copies are mixed by a mode-wise $50{:}50$ beam splitter, followed by occupation measurements. For every pure fixed-number input, $\Var(\Delta n)=2\nu$ [Eq.~\eqref{eq:casimirHOM}]. The full number-difference histogram additionally determines the replica-singlet weight $\A_2$ and hence the exact collision moment. (c)~The asymptotic collision randomness $2/\CN_N$ along the uniform four-mode family, plotted against defect density $\gamma=\nu/N$.
The open circle is the Slater endpoint, where $\CN_N=\Theta(N)$. At fixed $\gamma>0$, $\CN_N\to2/\gamma$, and the Porter--Thomas point
is reached only at $\gamma=1$.}
\label{fig:witness}
\end{figure}

Because $\E_VX_x=1$, $\CN_N$ measures the collision enhancement over
a flat distribution: $\CN_N=1$ for perfectly uniform probabilities. For a Haar-random state the benchmark is $2D_N/(D_N+1)\to2$, equivalently the second moment of
an exponential Porter--Thomas law. We call $\CN_N=O(1)$ collision
anticoncentration~\cite{Dalzell2022Random}. This second-moment
condition is stronger than the tail condition used in
sampling-hardness arguments~\cite{Hangleiter2018Anticoncentration,
Hangleiter2023Computational}: by Paley--Zygmund, it keeps
$\Pr_V[X_x\ge\alpha]$ bounded away from zero for every fixed
$0<\alpha<1$.
Unless stated otherwise, all asymptotic estimates below for four-mode products are uniform over system sizes and block parameters $0\le\lambda_i\le1/2$ with total defect
$4\sum_i\lambda_i(1-\lambda_i)=\nu$. Since $\nu\le N$, the limit $\nu\to\infty$ entails $N\to\infty$ while allowing $\gamma\to0$.

\smallskip
\paragraph*{Two copies on a beam splitter.}
Two-copy interference measures $\nu$ without reconstructing $\Gamma$. Take two identical copies of a pure fixed-number input and mix corresponding modes with identical beam splitters~\cite{Pichler2013Thermal} [Fig.~\ref{fig:witness}(b)]. This uniform mixing is unchanged by a common orbital-basis rotation. For a Slater determinant, every natural orbital is filled in both copies or in neither, so Pauli blocking forbids any transfer at any transmissivity. Fractional occupations allow transfer, and the defect quantifies the resulting number fluctuations. Let $J_+=\sum_jc_j^{(1)\dagger}c_j^{(2)}$ move a particle from copy $2$ to copy $1$, let $J_-=J_+^\dagger$ move it back, and let
$J_z=(\hat n^{(1)}-\hat n^{(2)})/2$, where $\hat n^{(a)}$ is the number operator of
copy $a$. Together with $J_x=(J_++J_-)/2$ and $J_y=(J_+-J_-)/(2i)$ these
operators form a replica angular momentum. The number-conserving Gaussian rotation $e^{-i\varphi J_y}$ is a beam splitter with transmissivity $\cos^2(\varphi/2)$. In the doubled input,
$\langle J_+J_-\rangle=\langle J_-J_+\rangle=\tr[\Gamma(\id-\Gamma)]=\nu$ (SM).

The doubled input has $J_z=0$, so its total replica angular momentum $J^2=J_+J_-+J_z^2-J_z$ has expectation value $\nu$. Rotation symmetry about the $z$ axis gives $\langle J_x\rangle=0$ and $\langle J_x^2\rangle=\langle J_y^2\rangle=\nu/2$. At $50{:}50$ ($\varphi=\pi/2$), the number difference $\Delta n=\hat n^{(1)}-\hat n^{(2)}$ is rotated into $-2J_x$, giving
\begin{equation}
  \nu=\matel{\Psi^{\otimes2}}{J^{2}}{\Psi^{\otimes2}}
  =\frac12\,\Var(\Delta n)\big|_{50{:}50}.
  \label{eq:casimirHOM}
\end{equation}
Related two-copy interference protocols have been realized with bosonic atoms~\cite{Islam2015Measuring,Kaufman2016Quantum}. Bell-sampling estimators access the same one-body defect~\cite{Haug2026Practical,Tarabunga2026Fermionic}.

The component unchanged by these rotations connects this readout to collision statistics. Let $\Pi_2$ project onto the replica-singlet sector and $\A_2(\Psi):=\matel{\Psi^{\otimes2}}{\Pi_2}{\Psi^{\otimes2}}$ be its weight. It equals $1$ exactly for Slater determinants, and $1-\A_2\le\nu/6$ controls its deficit at small $\nu$ (SM). The full $50{:}50$ number-difference histogram determines $\A_2$ and hence the exact collision moment below. The SM gives the reconstruction and treats mixed or mismatched copies.

\smallskip
\paragraph*{Collision statistics from two-copy interference.}
The second moment of a fixed outcome, $\E_V|\matel{x}{V}{\Psi}|^4$, is the expectation value in $\Psi^{\otimes2}$ of the Haar-averaged doubled projector $\E_V[(V^\dagger)^{\otimes2}
(\ket{x}\bra{x})^{\otimes2}V^{\otimes2}]$.
The doubled outcome $\ket{x}^{\otimes2}$ is a doubled Slater determinant and hence a singlet, and $V^{\otimes2}$, acting identically on both copies, commutes with all replica rotations. The singlet sector is irreducible and occurs with multiplicity one under physical mode
rotations (End Matter). Schur's lemma therefore gives the averaged projector $\kappa_N\Pi_2/D_N^2$~\cite{Oszmaniec2022Fermion}, where
$D_N^2/\kappa_N$ is the sector dimension and
$\kappa_N:=(2N+1)^{2}/(4N+1)=N\,[1+O(N^{-1})]$.

The singlet weight is itself an interference signal. The singlet projector is the average over replica $\SU(2)$ rotations. In the Euler decomposition, rotations about $z$ act trivially on the doubled input because $J_z=0$, leaving only the beam-splitter angle $\varphi$
\begin{equation}
  \begin{aligned}
  \A_2(\Psi)
  &=\frac12\int_0^{\pi}
  \mathcal F_{\Psi}(\varphi)\,\sin\varphi\,d\varphi,\\
  \mathcal F_{\Psi}(\varphi)
  &:=\matel{\Psi^{\otimes2}}{e^{-i\varphi J_y}}{\Psi^{\otimes2}}.
  \end{aligned}
  \label{eq:polaraverage}
\end{equation}
The return amplitude $\mathcal F_\Psi$ equals $1$ at $\varphi=0$ and again at $\varphi=\pi$: exchanging identical pure fixed-number copies returns the doubled state with cancelling fermionic phases (SM).

\begin{theorem}[Collision-moment identity]\label{thm:collision}
For every pure half-filled input $\Psi$,
\begin{equation}
  \CN_N(\Psi)=\kappa_N\,\A_2(\Psi),
  \qquad \kappa_N=\frac{(2N+1)^2}{4N+1}.
  \label{eq:order}
\end{equation}
For the four-mode inputs $\ket{\Psi_{\bm\lambda}}$, the exact block-profile integral is
\begin{equation}
  \CN_N(\bm\lambda)=\kappa_N
  \int_0^1
  \prod_{i=1}^{N}
  \Big(1-\lambda_i(1-\lambda_i)(1-z^{2})\Big)\,dz .
  \label{eq:integral}
\end{equation}
Here $z=|\cos\varphi|$ folds the two angular endpoints onto $z=1$.
For this four-mode family, uniformly over system sizes and block profiles,
\begin{equation}
  \CN_N=\frac{2}{\gamma}\,\big[1+O(\nu^{-1})\big]
  \qquad (\nu\to\infty).
  \label{eq:hardedge}
\end{equation}
The leading asymptotic remains uniform for products of blocks whose sizes are bounded independently of $N$ (SM).
\end{theorem}

\begin{figure}[t]
\centering
\includegraphics[width=1\columnwidth]{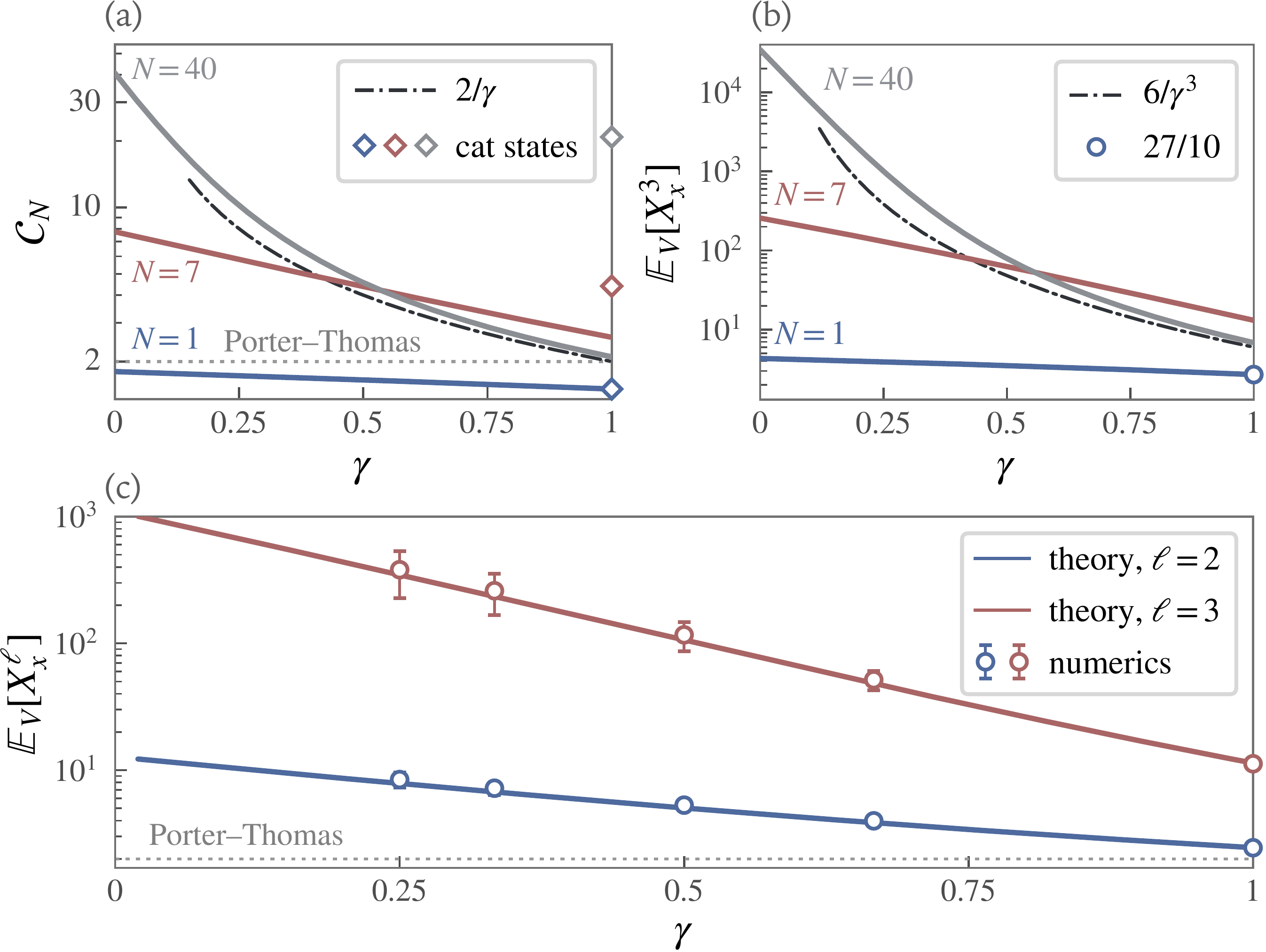}
\caption{\textbf{Collision and third moments for uniform profiles.}
(a) Exact collision moments and the asymptote $2/\gamma$; diamonds show complementary two-Slater cats with $\nu=N$. (b)~Third moments and the asymptote $6/\gamma^3$. The open circle marks the exact $N=1$, $\gamma=1$ value $27/10$ derived in the SM. (c) Theory at $N=12$ compared with direct Haar sampling; the dotted line marks the Porter--Thomas value $2$ for $\ell=2$; numerical protocols and uncertainty estimates are given in the SM.}
\label{fig:cost}
\end{figure}

Equation~\eqref{eq:order} compares the input weight in the selected two-copy sector with that sector's fraction of the doubled Hilbert space. For Haar mixing on the full $D_N$-dimensional sector, every doubled pure input has unit weight in the symmetric sector, whose fraction $(D_N+1)/(2D_N)\to1/2$ gives the value $2$. Passive FLO selects the smaller fraction $1/\kappa_N\sim1/N$, with all input dependence in $\A_2$. Slater inputs saturate the bound $\CN_N\le\kappa_N$. Collision anticoncentration therefore requires $\A_2=O(N^{-1})$, and matching the Porter--Thomas value specifically requires $\A_2\sim2/N$.

For the four-mode product inputs, only the two endpoint peaks contribute at leading order as $\nu\to\infty$. Near $\varphi=0$, $\mathcal F_\Psi(\varphi)=\exp[-\nu\varphi^2/4+O(\nu\varphi^4)]$, so the identity peak has angular width $O(\nu^{-1/2})$, and the polar measure $\sin\varphi\,d\varphi\simeq\varphi\,d\varphi$ turns this width into a contribution $1/\nu$ to $\A_2$. The exchange peak near $\varphi=\pi$ contributes the same. Together they give $\A_2\simeq2/\nu$ and hence $\CN_N\simeq2N/\nu$. At leading order only the total defect $\nu=4\sum_i\lambda_i(1-\lambda_i)$ enters. Slater blocks add none, and the defect may be spread over many weakly non-Gaussian blocks. Profile dependence first enters the relative $O(\nu^{-1})$ correction (SM). The two peaks are the $\ell=2$ precursors of the replica permutations that govern all fixed moments.

The asymptotic law~\eqref{eq:hardedge} holds whenever $\nu\to\infty$. For example, $\nu=\sqrt N$ gives $\CN_N\sim2\sqrt N$. Together with the necessary condition in Theorem~\ref{thm:noroute}, its bounded-block extension shows that these products anticoncentrate, $\CN_N=O(1)$, if and only if $\nu=\Theta(N)$. If the density has a positive limit, $\gamma\to\bar\gamma\in(0,1]$, then $\CN_N\to2/\bar\gamma$ [Fig.~\ref{fig:cost}(a)], reaching the Porter--Thomas value only at $\bar\gamma=1$. Paley--Zygmund then gives, for every fixed outcome, $\Pr_V[X_x\ge\alpha]\ge(\bar\gamma/2)(1-\alpha)^{2}+o(1)$ for $0<\alpha<1$. Its per-circuit form is given in SM Sec.~\ref{sm:typical}.

\smallskip
\paragraph*{All moments and the Porter--Thomas threshold.}
Collision anticoncentration bounds the second moment, whereas Porter--Thomas statistics require every fixed moment to approach $\ell!$. The two-copy picture extends to $\ell$ copies: the singlet sector of collective $\SU(\ell)$ rotations is spanned by $\ell$-fold copies of Slater determinants, and the physical Haar average spreads $\ket{x}^{\otimes\ell}$ uniformly over it (SM). Thus $\E_V[X_x^\ell]=R_\ell(N)\A_\ell(\Psi)$ for every
pure input, where $\A_\ell$ is the singlet weight and $R_\ell(N)$ is the normalized $\ell$th moment for a Slater input, with $R_2(N)=\kappa_N$. For four-mode products, $\A_\ell$ reduces to an integral over replica mixing angles (SM), giving the following law. Figure~\ref{fig:cost}(b),(c) shows the third moment
and direct-Haar checks.

\begin{theorem}[Porter--Thomas moment threshold]\label{thm:PT}
Fix an integer $\ell\ge2$. For the four-mode product family,
\begin{equation}
  \E_V\big[X_x^{\ell}\big]
  =
  \ell!\,
  \gamma^{-\binom{\ell}{2}}
  \big[1+o_\ell(1)\big]
  \qquad (\nu\to\infty),
  \label{eq:PTlaw}
\end{equation}
where the error depends on $\ell$. If $\gamma\to\bar\gamma\in(0,1]$, then
$\E_V[X_x^{\ell}]\to\ell!\,\bar\gamma^{-\binom{\ell}{2}}$. The fixed-outcome Porter--Thomas moment hierarchy, $\E_V[X_x^{\ell}]\to\ell!$ for all fixed $\ell$, holds if and only if $\gamma\to1$.
\end{theorem}

After phase rotations are removed, the replica integrand has $\ell!$ equivalent maxima, one for each permutation of the copies, generalizing the identity and exchange peaks. The factor $\ell!$ is the same permutation count as in the Wick contractions of a complex
Gaussian amplitude. Near each maximum, each replica pair supplies two mixing directions of width $O(\nu^{-1/2})$, contributing an area $1/\nu$ to $\A_\ell$ (SM). The physical Haar factor $R_\ell(N)$ supplies one factor $N$ per pair, leaving $N/\nu$ for each of the $\binom{\ell}{2}$ pairs. All Porter--Thomas moments are recovered only when $\nu/N\to1$.

For $0<\bar\gamma<1$ the limiting moments admit a simple probabilistic reading. They are the moments of an exponential variable multiplied by an independent mean-one lognormal factor with log-variance $\log(1/\bar\gamma)$ (SM). The hierarchy thus has a lognormal broadening signature set by the inverse defect density. We do not claim
convergence to this product law, because for $\bar\gamma<1$ the lognormal moment problem is indeterminate. At $\bar\gamma=1$ the broadening vanishes, the moments determine the law, and $X_x$ converges to an exponential variable.

\begin{figure}[t]
\centering
\includegraphics[width=1\columnwidth]{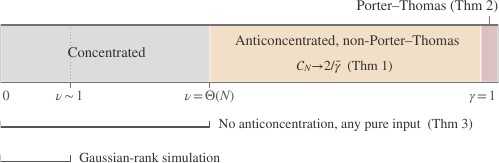}
\caption{\textbf{Sampling regimes versus defect density $\gamma=\nu/N$.}
Subextensive defect universally excludes anticoncentration, while bounded-block products anticoncentrate at extensive defect with $\CN_N\to2/\bar\gamma$ for $\gamma\to\bar\gamma>0$. The four-mode family reaches the Porter--Thomas hierarchy only as $\gamma\to1$. For $\nu=O(1)$, explicit truncated-rank descriptions also permit fixed-accuracy Gaussian-rank simulation (SM). The fermion-sampling input of Ref.~\cite{Oszmaniec2022Fermion} corresponds to $\gamma=1$.}
\label{fig:regimes}
\end{figure}

\smallskip
\paragraph*{A universal obstruction.}
The extensive threshold also constrains arbitrary pure inputs. Integer moments can be dominated by rare heavy outcomes. To test whether typical rescaled probabilities vanish, we instead bound fractional moments.

\begin{theorem}[Universal necessary condition for anticoncentration]\label{thm:noroute}
For every fixed $0<r<1$, there is a constant $C_r<\infty$ such that every pure half-filled input obeys
\begin{equation}
  \E_V\big[X_x^{r}\big]
  \le
  C_r\,
  \left(\frac{1+\nu}{N}\right)^{r(1-r)/2}.
  \label{eq:nogo}
\end{equation}
\end{theorem}

For $\nu=o(N)$, Markov's inequality gives $\Pr_V[X_x\ge\alpha]\to0$ for every fixed $\alpha>0$. Thus
anticoncentration requires $\nu=\Omega(N)$ without a product or low-entanglement assumption. Physically, subextensive defect leaves a natural-orbital Fermi sea dressed by an excitation cloud of mean size $O(\nu)$. Correlations within and between excitation sectors remain
unrestricted in the proof (SM). The scaling is sharp up to a constant: products of maximally non-Gaussian four-mode blocks padded with Slater blocks attain it as $\nu\to\infty$ (SM). Together with Theorem~\ref{thm:collision}, this establishes the extensive threshold for bounded-block products [Fig.~\ref{fig:regimes}].

Extensive defect is nevertheless not sufficient. The equal superposition of two Slater determinants filling complementary orbital sets has $\nu=N$ yet $\CN_N>N/2$ and $X_x\to0$ in probability (SM). Each doubled branch is a replica singlet, so at least half the two-copy weight remains Pauli blocked. More generally, an expansion
$\ket{\Psi}=\sum_{k=1}^{\chi}c_k\ket{f_k}$ in pairwise orthonormal Slater determinants gives $\CN_N(\Psi)\ge\kappa_N\sum_k|c_k|^4\ge\kappa_N/\chi$ (SM).
Collision anticoncentration therefore requires $\chi=\Omega(N)$ in every such expansion. The defect sums the natural-orbital occupation variances, but the collision moment also depends on many-body correlations.

\smallskip
\paragraph*{Preparation cost.}
The defect also bounds the resources needed to prepare $\Psi$. For a pure fixed-parity state, let $M(\Psi)$ be the von Neumann entropy in bits of the Gaussian state with the same full Majorana covariance matrix. For fixed-number states, pairing correlators vanish and this covariance is determined by $\Gamma$, whose eigenvalues $\theta_a$ give $M(\Psi)=\sum_a h_2(\theta_a)\ge4\nu(\Psi)$, where $h_2$ is the binary
entropy and $h_2(\theta)\ge4\theta(1-\theta)$~\cite{Tarabunga2026Computable,
Turkeshi2026Williamson}. The quantity $M$ is additive and does not increase on average under pure-state Gaussian protocols~\cite{Turkeshi2026Williamson}, whereas a parity-preserving gate on $q$ modes changes it by at most
$2q$ bits (SM).

Consider a deterministic exact preparation of a pure fixed-number target state $\Psi$. It consumes a joint pure fixed-parity resource state on $q_0$ modes, applies parity-preserving non-Gaussian gates on $q_1,q_2,\ldots$ modes, and uses pure-state Gaussian protocols between them. Then
\begin{equation}
 q_0+2\sum_{t\ge1}q_t\ge M(\Psi)\ge4\nu(\Psi).
 \label{eq:cost}
\end{equation}
Each resource mode contributes at most one bit and gate $t$ at most $2q_t$. The SWAP-only case recovers Ref.~\cite{Tarabunga2026Computable}. If $\nu(\Psi)\ge cN$, the bound stays linear for pure fixed-number approximations within fixed trace distance $\epsilon<c/4$ (SM). Combined with Theorem~\ref{thm:noroute}, Eq.~\eqref{eq:cost} requires
$q_0+2\sum_tq_t=\Omega(N)$ for every anticoncentrating pure half-filled target.

\smallskip
\paragraph*{Sampling implications.}
For the uniform four-mode family, exact amplitudes and probabilities are $\#\mathrm P$-hard in the worst case at every fixed algebraic defect density $0<\gamma\le1$ (SM), by embedding the determinant-sum construction of
Ref.~\cite{Ivanov2017Computational}. At $\gamma=1$, the density of the fermion-sampling input of Ref.~\cite{Oszmaniec2022Fermion}, assume that estimating a fixed output probability to relative error $1/4$ with joint success probability at least $2/3$ is $\#\mathrm P$-hard in the sense specified in the SM. The full Porter--Thomas limit then rules out a polynomial-time classical sampler achieving ensemble-averaged $\ell_1$ error at most $1/64$ for all sufficiently large $N$, unless the polynomial hierarchy collapses to its third level (SM, Corollary~\ref{cor:hardness}). Under the fixed-density analogue of this assumption, the moment hierarchy certifies an $\ell_1$ sampling-error tolerance that increases strictly with $\gamma$ throughout $9/10\le\gamma\le1$ (SM, Proposition~\ref{smprop:densitytolerance}). Here $\ell_1$ distance is twice the total-variation distance.

\smallskip
\paragraph*{Discussion and outlook.}
Passive fermions contrast with doped Clifford circuits: for pure half-filled inputs, a vanishing defect density cannot generate even anticoncentration. Within the uniform four-mode family, increasing $\gamma$ lowers the limiting collision moment $2/\gamma$ and the limiting moments $\ell!\gamma^{-\binom{\ell}{2}}$, which reach their Porter--Thomas values only at $\gamma=1$. Near $\gamma=1$, increasing $\gamma$ also raises the sampling error that the hardness argument tolerates. Anticoncentration therefore carries an extensive magic cost, quantified by the same defect that lower-bounds preparation resources and is measured by two-copy number fluctuations. The constituent preparation, interference, and readout operations connect this setting to fermionic tweezer and ultracold-atom platforms~\cite{GonzalezCuadra2023Fermionic,
Jain2026Programmable,Islam2015Measuring,Kaufman2016Quantum}.

The conjectured fractional-moment law for many weak blocks (SM) and Slater-superposition models proposed for interacting integrable eigenstates~\cite{Swietek2026OneBody} offer tests of how correlations shape fermionic randomness. Beyond product inputs, it remains to identify sufficient conditions on many-body correlations for Porter-Thomas statistics.

\smallskip
\paragraph{Acknowledgments.}
\begin{acknowledgments}
This work has been supported by the BMFTR (RealistiQ, QSolid, MuniQC-Atoms, QuSol, Hybrid++, FermiQP, PasQuops), the Munich Quantum Valley (K-8), the Quantum Flagship (Millenion, PasQuans2), Berlin Quantum, the European Research Council (DebuQC), the Clusters of Excellence MATH+ and ML4Q, the DFG (CRC 183 and SPP 2514), and the Alexander von Humboldt Foundation.
\end{acknowledgments}

\bibliographystyle{apsrev4-2}
\bibliography{refs}

\clearpage
\section*{End Matter}

\subsection{Replica reduction and the collision law}
Ref.~\cite{Oszmaniec2022Fermion} introduced the two-copy Schur projector behind Eq.~\eqref{eq:order}. Related fermionic commutants and bosonic moment transitions appear in
Refs.~\cite{Braccia2026Commutant,Ehrenberg2025Transition}.
On two copies, the commuting physical and replica actions give
\begin{equation}
  \wedgepow{2N}\CC^{4N}\otimes\wedgepow{2N}\CC^{4N}
  \cong
  \bigoplus_{a=0}^{2N}
  \Schur_{(2^{a},1^{4N-2a})}\CC^{4N}.
  \label{eq:emdecomp}
\end{equation}
Each summand is the irreducible $\U(4N)$ module labelled by the indicated partition.
The doubled outcome is a replica singlet and selects the rectangle $a=2N$. Schur's lemma and the rectangular hook-content formula give
\begin{equation}
  \frac{D_N^2}{\dim\Schur_{(2^{2N})}\CC^{4N}}
  =\frac{(2N+1)^2}{4N+1}=\kappa_N.
  \label{eq:emdim}
\end{equation}
This proves Eq.~\eqref{eq:order}.

For $\ell=2$, replica averaging reduces to Eq.~\eqref{eq:polaraverage}. One canonical four-mode block has return amplitude
\begin{equation}
  \mathcal F_{\psi_\lambda}(\varphi)
  =1-\lambda(1-\lambda)\sin^2\varphi.
  \label{eq:emfringe}
\end{equation}
Multiplying the block overlaps and folding with $z=|\cos\varphi|$ gives Eq.~\eqref{eq:integral}.

With $\tau=\nu(1-z)$, the logarithm of the product is
$-\tau/2+O(\tau^2/\nu)$ and dominated convergence gives
\begin{equation}
  \nu\int_0^1\prod_i
  \left[1-\lambda_i(1-\lambda_i)(1-z^2)\right]dz
  \longrightarrow2,
\end{equation}
so $\CN_N\sim2N/\nu$ (Fig.~\ref{fig:maxmagic}). Section~\ref{sm:collision} proves the uniform relative error $O(\nu^{-1})$ and the bounded-block extension.

\subsection{Proof reduction for the universal obstruction}
\label{em:nogo}

In this subsection, we reduce the universal bound to excitation-sector weights and a uniform estimate for their Haar-random overlaps.
Set $n=2N$ and fix $0<r<1$ with
$\beta=r(1-r)/2$. For $\rho=\ket\Psi\bra\Psi$, let $F$ span the $n$ most occupied natural orbitals and $G=F^\perp$.
The sector $\mathcal E_j=\wedgepow{n-j}F\otimes\wedgepow jG$
has projector $Q_j$ and dimension $d_j=\binom nj^2$.
Its weight $w_j=\tr(\rho Q_j)$ obeys $\mu:=\sum_jjw_j\le2\nu$ [Eq.~\eqref{eq:smmunu}].
Thus the defect bounds the mean number of excitations out of $F$.

Average $\rho$ over $H=\U(F)\times\U(G)$. Concavity of the
fractional moment and Haar invariance give
\begin{equation}
 \begin{aligned}
  \Tw_H\rho&=\sum_{j=0}^n w_j\frac{Q_j}{d_j},\\
  \E_V[X_x^r]&\le\E_P\left(\sum_{j=0}^n w_jZ_j(P)\right)^r,
 \end{aligned}
 \label{eq:emtwirlsectors}
\end{equation}
where $P$ is Haar on $\Gr(n,2n)$ and
$Z_j(P)=D_N\matel{P}{Q_j}{P}/d_j$, with $\E_PZ_j=1$.
No restriction on correlations within or between sectors is needed. The key estimate, proved in Lemma~\ref{smlem:dyadic}, is
\begin{equation}
 \E_P\!\left[\sup_{L\le j<2L}Z_j^r\right]
 \le C_r(L/n)^{\beta},\qquad 1\le L\le n/8.
 \label{eq:emdyadic}
\end{equation}
For dyadic $L<n/8$, the sectors with $L\le j<\min(2L,n/8)$ have total weight at most $\min(1,\mu/L)$. The $j=0$ contribution is bounded by the Slater
law $F_r(n)\le C_rn^{-\beta}$ [Eq.~\eqref{eq:smfreelaw}].
The sectors $j\ge n/8$ have total weight at most $\min(1,8\mu/n)$, and Jensen's inequality bounds their contribution by $C_r(\mu/n)^r$. Subadditivity therefore gives
\begin{equation}
 \begin{aligned}
  \E_V[X_x^r]&\le C_r\bigg[n^{-\beta}+(\mu/n)^r\\
   &\quad+\sum_L\min(1,\mu/L)^r(L/n)^{\beta}\bigg]\\
  &\le C_r'\left(\frac{1+\mu}{n}\right)^{\beta}.
 \end{aligned}
 \label{eq:emnogoreduction}
\end{equation}
Splitting the sum at $L\simeq1+\mu$ gives two geometric sums
because $0<\beta<r$. For $n<8$, Jensen's bound
$\E_V[X_x^r]\le1$ is absorbed into $C_r$. Finally, $\mu\le2\nu$ and $n=2N$ yield Theorem~\ref{thm:noroute}.

\begin{figure}[!t]
\centering
\includegraphics[width=0.85 \columnwidth]{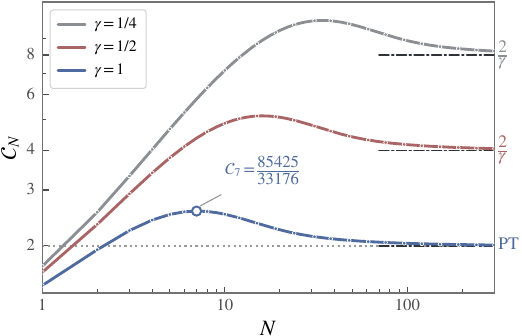}
\caption{\textbf{Uniform-profile collision moments.}
Exact $\CN_N$ for the uniform family
$4\lambda(1-\lambda)=\gamma$; dash-dotted: $2/\gamma$; the
$\gamma=1$ maximum occurs at $N=7$.}
\label{fig:maxmagic}
\end{figure}

\clearpage
\onecolumngrid

\begin{center}
\textbf{\large Supplemental Material for
``The magic cost of fermionic randomness''}
\end{center}

\setcounter{equation}{0}
\setcounter{figure}{0}
\setcounter{table}{0}
\setcounter{section}{0}
\setcounter{theorem}{0}
\setcounter{proposition}{0}
\setcounter{lemma}{0}
\setcounter{corollary}{0}
\setcounter{remark}{0}
\setcounter{conjecture}{0}
\renewcommand{\theequation}{S\arabic{equation}}
\renewcommand{\thefigure}{S\arabic{figure}}
\renewcommand{\theHfigure}{S\arabic{figure}}
\renewcommand{\thesection}{S\arabic{section}}
\makeatletter
\let\c@proposition\c@theorem
\let\c@lemma\c@theorem
\let\c@corollary\c@theorem
\let\c@conjecture\c@theorem
\let\c@remark\c@theorem
\makeatother
\renewcommand{\thetheorem}{S\arabic{theorem}}
\renewcommand{\theproposition}{S\arabic{theorem}}
\renewcommand{\thelemma}{S\arabic{theorem}}
\renewcommand{\thecorollary}{S\arabic{theorem}}
\renewcommand{\theconjecture}{S\arabic{theorem}}
\renewcommand{\theremark}{S\arabic{theorem}}

\renewcommand{\theHequation}{supp.\arabic{equation}}
\renewcommand{\theHsection}{supp.\arabic{section}}
\providecommand{\theHtheorem}{\arabic{theorem}}
\renewcommand{\theHtheorem}{supp.\arabic{theorem}}
\providecommand{\theHproposition}{\arabic{proposition}}
\renewcommand{\theHproposition}{supp.\arabic{theorem}}
\providecommand{\theHlemma}{\arabic{lemma}}
\renewcommand{\theHlemma}{supp.\arabic{theorem}}
\providecommand{\theHcorollary}{\arabic{corollary}}
\renewcommand{\theHcorollary}{supp.\arabic{theorem}}
\providecommand{\theHremark}{\arabic{remark}}
\renewcommand{\theHremark}{supp.\arabic{theorem}}
\providecommand{\theHconjecture}{\arabic{conjecture}}
\renewcommand{\theHconjecture}{supp.\arabic{theorem}}

\vspace{4pt}
\noindent
Throughout,
$4N$ modes at half filling host $n=2N$ particles,
$D_N=\binom{4N}{2N}$, $\kappa_N=(2N+1)^2/(4N+1)$, $X_x=D_N\,p_V(x)$ for
$V=\wedgepow{2N}u$, and
$\lVert p-q\rVert_1=\sum_x|p_x-q_x|$ for outcome distributions.
Replica copies are embedded in the fermionic
Fock space $\wedgepow{\bullet}(\CC^{4N}\otimes\CC^{\ell})$ over the
mode--replica pairs, so operators in distinct replicas anticommute and
$\Psi^{\otimes\ell}$ denotes the embedded product state.  For the four-mode product family, we use the block defects
$m_i:=4\lambda_i(1-\lambda_i)$, so $\nu=\sum_i m_i$ and $\gamma=\nu/N$. The uniformity convention of the main text remains in force, and subscripts on error terms record dependence on a fixed replica number or block-size bound.

\section{Canonical blocks and replica moment reduction}
\label{sm:replica}

The canonical form of a four-mode block isolates the non-Gaussian input parameter, while replica duality converts Haar-averaged output moments into sector weights. We first establish this reduction and then connect its two-copy instance to measurable number fluctuations, the singlet weight, and the collision law.
By the Youla decomposition, some $W\in\U(4)$ maps every pure
two-fermion state on four modes to $\ket{\psi_\lambda}$ of
Eq.~\eqref{eq:blocknormal}. Phases can be absorbed into the orbital
pairs, which can be ordered so that $0\le\lambda\le1/2$. Its natural
occupations are $\{\lambda,\lambda,1-\lambda,1-\lambda\}$, so
$\tr(\Gamma-\Gamma^2)=4\lambda(1-\lambda)=m$. For block products,
$\bigoplus_iW_i$ is absorbed into the subsequent Haar unitary, and the
output law depends on block $i$ only through $m_i$.

Skew Howe duality~\cite{Howe1995Perspectives} makes
$\ket{x}^{\otimes\ell}$ select, with multiplicity one, the rectangle
$\Schur_{(\ell^{2N})}\CC^{4N}$. Applying the hook-content formula gives
\begin{equation}
  \begin{aligned}
  \dim\Schur_{(\ell^{2N})}\CC^{4N}
  &=
    \prod_{j=1}^{\ell}
    \frac{(4N+j-1)!\,(j-1)!}{[(2N+j-1)!]^{2}},\\
  R_\ell(N)
  &:=
    \frac{D_N^{\ell}}
         {\dim\Schur_{(\ell^{2N})}\CC^{4N}},\\
  R_\ell(N)
  &=
    \frac{N^{\binom{\ell}{2}}}{\GB(\ell+1)}
    [1+O_\ell(N^{-1})],\\
  R_2(N)
  &=\frac{(2N+1)^2}{4N+1}=\kappa_N.
  \end{aligned}
  \label{eq:smRq}
\end{equation}
Here $\GB$ is the Barnes $G$-function, with
$\GB(\ell+1)=\prod_{j=1}^{\ell-1}j!$.
With projector $\Pi_\ell$ and
$\A_\ell(\Psi)=\matel{\Psi^{\otimes\ell}}{\Pi_\ell}{\Psi^{\otimes\ell}}$,
Schur averaging gives
$\E_V[X_x^\ell]=R_\ell(N)\A_\ell(\Psi)$.
A Slater determinant has $\A_\ell=1$, since every inter-replica
hopping operator is Pauli blocked. Thus $R_\ell(N)$ is its exact
normalized $\ell$th moment, and $\ell=2$ gives Eq.~\eqref{eq:order}.

\subsection{Casimir identity and the singlet deficit}
\label{sm:certificate}

The replica Casimir relates one-body occupations to two-copy interference and quantifies the departure from the singlet sector.

\begin{proposition}[Two-copy certificate]\label{smprop:certificate}
For every pure half-filled input,
\begin{equation}
  \A_2(\Psi)=1
  \quad\Longleftrightarrow\quad
  \nu(\Psi)=0
  \quad\Longleftrightarrow\quad
  \Psi\ \hbox{is a Slater determinant},
\end{equation}
and $1-\A_2(\Psi)\le\nu/6$.  For the four-mode product family,
$\nu/6-\nu^2/60\le1-\A_2\le\nu/6$, with $1-\A_2=m/6$ for one block.
\end{proposition}

\begin{proof}
The doubled state has zero replica charge. Grouping $J_+J_-$ by copy
requires two fermionic interchanges. The resulting even operators
factorize between the independent copies, giving
\begin{equation}
  \begin{aligned}
  \langle J_+J_-\rangle
  &=\sum_{a,b}\Gamma_{a,b}(\delta_{a,b}-\Gamma_{ba})
    =\tr(\Gamma-\Gamma^2)=\nu,
  &\qquad \langle J_-J_+\rangle&=\nu,\\
  \langle J^2\rangle&=\nu,
  &\langle J_x^2\rangle=\langle J_y^2\rangle&=\frac{\nu}{2}.
  \end{aligned}
  \label{eq:emcasimir}
\end{equation}

Let $\omega_j$ be the weight of total replica spin $j$ in
$\ket{\Psi}^{\otimes2}$. On the zero-weight subspace,
$e^{-i\pi J_y}$ acts in spin $j$ as $(-1)^j$. This rotation sends
$c_a^{(1)\dagger}\mapsto c_a^{(2)\dagger}$ and
$c_a^{(2)\dagger}\mapsto-c_a^{(1)\dagger}$.
For $n$ particles per copy, the rotation and the return to copy order
contribute $(-1)^n$ and $(-1)^{n^2}$, respectively, whose product is $1$.
The embedded double is therefore invariant, so $\omega_j=0$ for odd $j$. Since
$\A_2=\omega_0$,
\begin{equation}
  \nu=\sum_{j\ {\rm even}}j(j+1)\omega_j
  \ge6(1-\A_2).
  \label{eq:emspin}
\end{equation}
It follows that $\nu=0$ if and only if the doubled state is a singlet.
For a pure fixed-number state, $\nu=0$ means $\Gamma^2=\Gamma$ and
hence characterizes a Slater determinant~\cite{Coleman1963Structure}.

For a block product, Eq.~\eqref{eq:emfringe} and the first two
Bonferroni inequalities give, pointwise in $\varphi$,
\begin{equation}
  \begin{aligned}
  \sum_i\frac{m_i}{4}\sin^2\varphi
  -\sum_{i<j}\frac{m_im_j}{16}\sin^4\varphi
  &\le
  1-\prod_i\left(1-\frac{m_i}{4}\sin^2\varphi\right)\\
  &\le
  \sum_i\frac{m_i}{4}\sin^2\varphi .
  \end{aligned}
  \label{eq:smblockbonferroni}
\end{equation}
Polar averaging uses
\begin{equation}
  \begin{aligned}
  \frac12\int_0^\pi
  \sum_i\frac{m_i}{4}\sin^3\varphi\,d\varphi
  &=\frac{\nu}{6},\\
  \frac12\int_0^\pi
  \sum_{i<j}\frac{m_im_j}{16}\sin^5\varphi\,d\varphi
  &=\frac1{30}\sum_{i<j}m_im_j
  \le\frac{\nu^2}{60},
  \end{aligned}
  \label{eq:smblockpolar}
\end{equation}
and hence
\begin{equation}
  \frac{\nu}{6}-\frac{\nu^2}{60}
  \le1-\A_2(\Psi_{\bm\lambda})
  \le\frac{\nu}{6}.
  \label{eq:emtwosided}
\end{equation}
For one block, both inequalities in
Eq.~\eqref{eq:smblockbonferroni} are equalities, and $1-\A_2=m/6$.
\end{proof}

The same spin decomposition also gives the complementary bound
\begin{equation}
 \frac{\nu}{2N(2N+1)}\le1-\A_2(\Psi)\le\frac{\nu}{6},
 \label{eq:jequantfaithfulness}
\end{equation}
since the maximal allowed replica spin is $2N$.
The lower bound is generally much weaker than the product-input asymptotic, but it makes the finite-size faithfulness quantitative.

The number-resolved beam-splitter readout follows from the same spin
weights, which give
\begin{equation}
  \begin{aligned}
  \mathcal F_\Psi(\varphi)
  &=1-\frac{\nu}{4}\varphi^2+O(\varphi^4),\\
  \Var(\Delta n)\big|_{50{:}50}
  &=4\langle J_x^2\rangle=2\nu,\\
  \Pr(\Delta n=2k)
  \big|_{50{:}50}
  &=\sum_{j\ge|k|}\omega_j|d^j_{k0}(\pi/2)|^2 .
  \end{aligned}
  \label{eq:smhistogram}
\end{equation}
Here $d^j_{k0}(\varphi):=\langle j,k|e^{-i\varphi J_y}|j,0\rangle$
is a Wigner small-$d$ matrix element in the $J_z$ basis.
The splitter rotates $\Delta n=2J_z$ to $-2J_x$.
Since
$|d^j_{j0}(\pi/2)|^2=\binom{2j}{j}4^{-j}>0$, descending in $|k|$
recovers the $\omega_j$ from this triangular system, including
$\A_2=\omega_0$. Because $\omega_j=0$ for odd $j$ and
$d^j_{k0}(\pi/2)=0$ for odd $j+k$, only even $k$ occur and the
histogram of identical pure copies is supported on
$\Delta n\in4\mathbb Z$. Weight at $\Delta n\equiv2\pmod{4}$
therefore certifies a copy mismatch or a loss of purity.
This conclusion assumes the fixed-number preparation and ideal number-resolved readout used here.
Invertibility establishes exact identifiability of $\A_2$. The bounded estimator in Sec.~\ref{sm:singletestimator} additionally gives an explicit number of shots sufficient for a prescribed additive accuracy.

For two uncorrelated fixed-number inputs, copy factorization gives
\begin{equation}
  \begin{aligned}
  \Var(\Delta n)\big|_{50{:}50}
  &=\langle J_+J_-+J_-J_+\rangle\\
  &=\tr[\Gamma_1(\id-\Gamma_2)]
    +\tr[\Gamma_2(\id-\Gamma_1)]\\
  &=\nu_1+\nu_2+\|\Gamma_1-\Gamma_2\|_{\rm F}^2,
  \qquad
  \nu_a=\tr(\Gamma_a-\Gamma_a^2).
  \end{aligned}
  \label{eq:smmixedvariance}
\end{equation}

\subsection{A bounded estimator of the replica-singlet weight}
\label{sm:singletestimator}

Reconstructing every spin weight is unnecessary if only the collision moment is needed. A bounded function of each measured number difference directly estimates the singlet weight. Its shot bound also specifies the precision cost of converting this measurement into a collision estimate.

\begin{proposition}[Number-difference estimator]
\label{smprop:singletestimator}
Let $\Omega$ be a normalized doubled density operator supported on the zero-charge subspace $J_z=0$, with total particle number $4N$. After a balanced beam splitter, write $K:=\Delta n/2$ and define
\begin{equation}
 h(k):=\begin{cases}
 (1-k^2)^{-1},&k\in2\mathbb Z,\\
 0,&k\in2\mathbb Z+1.
 \end{cases}
 \label{eq:jesingletfunction}
\end{equation}
Then
\begin{equation}
 \E_\Omega h(K)=\tr(\Pi_2\Omega).
 \label{eq:jesingletunbiased}
\end{equation}
In particular, for two identical pure half-filled input states the expectation equals $\A_2(\Psi)$. For $M$ independent shots, the sample mean $\widehat{\A}_2=M^{-1}\sum_{a=1}^M h(K_a)$ obeys
\begin{equation}
 \Pr\!\left(\left|\widehat{\A}_2-\tr(\Pi_2\Omega)\right|\ge\epsilon\right)
 \le2\exp\!\left(-\frac{9M\epsilon^2}{8}\right).
 \label{eq:jesingletshots}
\end{equation}
\end{proposition}

\begin{proof}
Decompose the zero-charge density operator into replica-spin sectors and let $\omega_j$ be their weights. Rotation and number measurement preserve total spin, so coherences between distinct spin sectors do not affect the outcome probabilities. Within each spin-$j$ sector, the zero-charge vector has magnetic quantum number zero. Consequently the characteristic function of $K$, after the balanced splitter, is
\begin{equation}
 \chi(\theta):=\E_\Omega e^{i\theta K}
 =\sum_{j=0}^{2N}\omega_j P_j(\cos\theta).
 \label{eq:jesingletcharacteristic}
\end{equation}
Here the sign of the splitter rotation is immaterial, since the number-difference distribution is symmetric. Legendre orthogonality gives
\begin{equation}
 \frac12\int_0^\pi\chi(\theta)\sin\theta\,d\theta
 =\omega_0=\tr(\Pi_2\Omega).
\end{equation}
On the other hand, direct integration yields
\begin{equation}
 \frac12\int_0^\pi\sin\theta\cos(k\theta)\,d\theta=h(k).
\end{equation}
The outcome support is finite, so exchanging the sum and integral proves Eq.~\eqref{eq:jesingletunbiased}. Finally, $-1/3\le h(k)\le1$, and Hoeffding's inequality for this interval gives Eq.~\eqref{eq:jesingletshots}.
\end{proof}

Thus
\begin{equation}
 M\ge\frac{8}{9\epsilon^2}\log\frac{2}{\delta}
\end{equation}
is sufficient for additive singlet-weight error at most $\epsilon$ with failure probability at most $\delta$, independently of $N$. For identical pure copies, $\widehat{\CN}_N=\kappa_N\widehat{\A}_2$ is an unbiased estimator of the normalized collision moment. Additive collision error $\epsilon_{\rm C}$ instead requires
\begin{equation}
 M\ge\frac{8\kappa_N^2}{9\epsilon_{\rm C}^2}\log\frac{2}{\delta},
 \label{eq:jecollisionshots}
\end{equation}
a sufficient bound of order $N^2$ at fixed accuracy and confidence. The estimator also applies to mixed or mismatched copies with the same fixed particle number; in that case it measures $\tr(\Pi_2\Omega)$ rather than the pure-input quantity $\A_2(\Psi)$.

\subsection{Collision law and endpoint asymptotics}
\label{sm:collision}
In this subsection, we discuss in more detail the
collision law and the concomitant endpoint asymptotics.

The two endpoints each contribute $\nu^{-1}[1+o(1)]$ to the pure-state
polar integral. In the blockwise dephasing model of
Sec.~\ref{sm:dephasing}, the exchange saddle has height $\tr\rho^2$ and
a separately modified curvature.
At fixed $\nu$, the logarithm of the product in
Eq.~\eqref{eq:integral} is concave in the profile.  The uniform
profile $m_i=\nu/N$ therefore maximizes the integral, while maximal
concentration minimizes it, with
$m_1=\cdots=m_{\lfloor\nu\rfloor}=1$ and the remainder, if nonzero, on
one block. Since
$-4x/3\le\log(1-x)\le-x$ on $[0,1/4]$,
\begin{equation}
  e^{-\nu(1-z^2)/3}
  \le
  \prod_i\left[1-\frac{m_i}{4}(1-z^2)\right]
  \le
  e^{-\nu(1-z^2)/4}.
  \label{eq:smbracket}
\end{equation}
Thus $\CN_N=\Theta(N/(1+\nu))$ with absolute constants.

\begin{proposition}[Collision endpoint asymptotic]\label{smprop:hardedge}
For every $N$ and profile $0\le m_i\le1$ with
$\nu=\sum_i m_i\ge1$,
\begin{equation}
  \begin{aligned}
  \int_0^1
  \prod_i\left[1-\frac{m_i}{4}(1-z^2)\right]dz
  &=\frac{2}{\nu}\big[1+O(\nu^{-1})\big],\\
  \CN_N&=\frac{2N}{\nu}\big[1+O(\nu^{-1})\big] .
  \end{aligned}
  \label{eq:smsharp}
\end{equation}
The implicit constants are absolute.
\end{proposition}

\begin{proof}
Set $\tau=\nu(1-z)$. The integral in
Eq.~\eqref{eq:integral} becomes
\begin{equation}
  \frac1\nu\int_0^\nu
  \exp\!\left\{
    \sum_i\log\!\left[
      1-\frac{m_i}{4}\frac{\tau}{\nu}
      \left(2-\frac{\tau}{\nu}\right)
    \right]
  \right\}\,d\tau .
  \label{eq:smrescaledendpoint}
\end{equation}
For $0\le\tau\le\nu$,
\begin{equation}
  \begin{aligned}
  \sum_i\frac{m_i}{4}\frac{\tau}{\nu}
      \left(2-\frac{\tau}{\nu}\right)
  &=\frac{\tau}{2}-\frac{\tau^2}{4\nu},\\
  \sum_i\left[
      \frac{m_i}{4}\frac{\tau}{\nu}
      \left(2-\frac{\tau}{\nu}\right)
    \right]^2
  &\le\frac{\tau^2}{4\nu},
  \end{aligned}
  \label{eq:smendpointmoments}
\end{equation}
and $|\log(1-x)+x|\le2x^2/3$ on $[0,1/4]$ shows that the exponent in
Eq.~\eqref{eq:smrescaledendpoint} differs from $-\tau/2$ by at most
$\tau^2/(2\nu)$.  The same exponent is at most $-\tau/4$, by
$\log(1-x)\le-x$ and $1-z^2\ge1-z$.
Writing $F(\tau)$ for this exponent, the mean value theorem gives
\begin{equation}
  \left|e^{F(\tau)}-e^{-\tau/2}\right|
  \le\frac{\tau^2}{2\nu}e^{-\tau/4}.
\end{equation}
Integrating this bound and extending the reference integral to infinity gives
\begin{equation}
  \int_0^\nu e^{F(\tau)}\,d\tau
  =\int_0^\infty e^{-\tau/2}d\tau+O(\nu^{-1})
  =2+O(\nu^{-1}),
\end{equation}
establishing the first half of Eq.~\eqref{eq:smsharp}. The factor
$\kappa_N=N[1+O(N^{-1})]$ contributes only
$O(\nu^{-1})$ to the relative error because $\nu\le N$.  The explicit
first correction is Eq.~\eqref{eq:smfirstcorr}.
\end{proof}

\smallskip
\paragraph*{Bounded blocks.}
Let $\CC^{4N}=\bigoplus_s\mathcal H_s$ with
$\dim\mathcal H_s\le b_*$, and let
$\Psi=\bigotimes_s\psi_s$ be a product of local states of fixed particle number
with total particle number $2N$.  Put
$\nu_s=\tr(\Gamma_s-\Gamma_s^2)$, so that $\nu=\sum_s\nu_s$.  If
$\omega_{s,j}$ are the replica-spin weights of $\psi_s^{\otimes2}$,
then
\begin{equation}
 \begin{aligned}
  \mathcal F_\Psi(\varphi)
  &=\prod_s\mathcal F_s(\varphi),\\
  \mathcal F_s(\varphi)
  &=\sum_{\substack{0\le j\le\dim\mathcal H_s/2\\j\ {\rm even}}}
    \omega_{s,j}P_j(\cos\varphi),\\
  \sum_j\omega_{s,j}&=1,
  &\nu_s&=\sum_jj(j+1)\omega_{s,j}.
 \end{aligned}
\end{equation}
Only even spins occur: for $k_s$ particles, the phase $(-1)^{k_s}$
from the $\pi$ replica rotation cancels the fermionic reordering sign
$(-1)^{k_s^2}$. Here $P_j$ is the Legendre polynomial.
Choose $\varphi_0>0$ sufficiently small depending only on $b_*$. For
the finitely many allowed nonzero even spins, compactness gives a
positive lower bound on $1-|P_j(\cos\varphi)|$ for
$\varphi_0\le\varphi\le\pi-\varphi_0$. Moreover,
\begin{equation}
 \nu_s\le
 \frac{b_*}{2}\left(\frac{b_*}{2}+1\right)(1-\omega_{s,0}),
\end{equation}
so this gap is uniform in the local state. Together with the uniform
endpoint Taylor expansions, this gives
\begin{equation}
  \mathcal F_s(\varphi)=1-\frac{\nu_s}{4}\varphi^2+O_{b_*}(\nu_s\varphi^4),
  \qquad
  |\mathcal F_s(\varphi)|\le
  \begin{cases}
    e^{-c_{b_*}\nu_s\varphi^2},
      &0\le\varphi\le\varphi_0,\\
    e^{-c_{\varphi_0,b_*}\nu_s},
      &\varphi_0\le\varphi\le\pi-\varphi_0,
  \end{cases}
  \label{eq:smboundedendpoint}
\end{equation}
and $\mathcal F_s(\pi-\varphi)=\mathcal F_s(\varphi)$.  Since $\nu_s\le b_*/4$, one has
$\sum_s\nu_s^2=O_{b_*}(\nu)$, and for fixed $y$
\begin{equation}
  \sum_s\log \mathcal F_s(y/\sqrt{\nu})
  =-\frac{y^2}{4}+O_{b_*}(y^4/\nu).
\end{equation}
Both endpoints of the polar integral contribute equally, while the
interior is $O(e^{-c\nu})$.  Rescaling
$y=\sqrt{\nu}\,\varphi$ gives
\begin{equation}
  \A_2(\Psi)=\frac{2}{\nu}[1+o_{b_*}(1)],
  \qquad
  \CN_N(\Psi)=\frac{2N}{\nu}[1+o_{b_*}(1)],
  \label{eq:smboundedhardedge}
\end{equation}
as $\nu\to\infty$. For fixed $b_*$, the error depends only on $\nu$
and is uniform over $N$, the block decomposition, and the local states.

\smallskip
\paragraph*{Finite-size forms.}
Let $\bar m:=\nu^{-1}\sum_i m_i^2$ be the weighted mean of the profile.
Expanding the same endpoint integral
one order further gives
\begin{equation}
  \frac{\nu\CN_N}{2\kappa_N}
  =1+\frac{2-\bar m}{\nu}+O(\nu^{-2})
  \label{eq:smfirstcorr}
\end{equation}
as $\nu\to\infty$.

For $m_i=\gamma$,
\begin{equation}
  \int_0^1
  \left[1-\frac{\gamma}{4}(1-z^2)\right]^Ndz
  =
  {}_2F_1\!\left(-N,1;\frac32;\frac{\gamma}{4}\right).
  \label{eq:sm2F1}
\end{equation}
At $\gamma=1$ this becomes
\begin{equation}
  \CN_N
  =\kappa_N I_N,
  \qquad
  I_N=\int_0^1\left(\frac{3+z^2}{4}\right)^Ndz,
  \qquad
  \sup_N\CN_N=\CN_7.
  \label{eq:smmaxmagic}
\end{equation}
Integrating the derivative of $z[(3+z^2)/4]^N$ gives
\begin{equation}
  (2N+1)I_N=1+\frac{3N}{2}I_{N-1}.
  \label{eq:smrec}
\end{equation}
Remark~4 of Ref.~\cite{Oszmaniec2022Fermion} suggested numerically,
through the combinatorial upper estimate in its Eq.~(E73), the uniform
bound $N\A_2\le2.4$. Retaining the exact Schur-sector dimension, the
Paley--Zygmund denominator is $\CN_N=\kappa_N\A_2$. Set
$\Lambda_N:=8N^3-18N^2+9N+2$ and $\Upsilon_N:=16N^2-4N-6$.  The
recursion \eqref{eq:smrec} gives
\begin{equation}
  \CN_N-\CN_{N-1}
  =-\frac{\Lambda_NI_{N-1}-\Upsilon_N}{2(4N-3)(4N+1)},
\end{equation}
so the sequence decreases at step $N$ exactly when
$I_{N-1}>\Upsilon_N/\Lambda_N$.
Iterating Eq.~\eqref{eq:smrec} from $I_0=1$ gives
$I_7=1139/3432>493/1509=\Upsilon_8/\Lambda_8$, and the threshold
propagates. If $I_{N-1}>\Upsilon_N/\Lambda_N$ for some $N\ge8$, then
\begin{equation}
  \begin{aligned}
  I_N
  =\frac{1+\frac{3N}{2}I_{N-1}}{2N+1}
  &>\frac{1+\frac{3N}{2}\,\Upsilon_N/\Lambda_N}{2N+1}\\
  &=\frac{\Upsilon_{N+1}}{\Lambda_{N+1}}
   +\frac{2(4N+1)\left(56N^3-4N^2-50N-5\right)}
         {(2N+1)\,\Lambda_N\Lambda_{N+1}}
  >\frac{\Upsilon_{N+1}}{\Lambda_{N+1}}.
  \end{aligned}
\end{equation}
Induction gives strict decrease for $N\ge8$, while direct iteration gives
the increase through $N=7$ and hence
$\sup_{N\ge1}\CN_N=\CN_7=85425/33176$ (Fig.~\ref{fig:maxmagic}).

\section{Orthogonal Slater decompositions and the cat-state obstruction}
\label{sm:purity}

Extensive one-body defect is necessary for anticoncentration, but it is not sufficient for arbitrary input states. Orthogonal Slater decompositions retain a minimum singlet weight through Pauli blocking. A complementary Slater cat state illustrates the resulting obstruction even at maximal defect density.
Pauli blocking already gives a lower bound on the singlet weight.  If
$\ket{\Psi}=\sum_kc_k\ket{f_k}$ is an expansion in pairwise orthonormal
Slater determinants, then each $\ket{f_k}^{\otimes2}$ is a replica
singlet, since $J_z$ vanishes 
on it and $J_+$ has no empty orbital to
hop into.  These doubled determinants are orthonormal, so
$\Pi_2\ge\sum_k\ket{f_k}^{\otimes2}\bra{f_k}^{\otimes2}$, and
evaluating on $\ket{\Psi}^{\otimes2}$ gives
\begin{equation}
  \A_2(\Psi)\ge\sum_k|c_k|^4,\qquad
  \CN_N(\Psi)\ge
  \kappa_N\sum_k|c_k|^4 ,
  \label{eq:smpurity}
\end{equation}
the second inequality by Eq.~\eqref{eq:order}.

By Cauchy--Schwarz, an expansion with $\chi$ terms satisfies
$\sum_k|c_k|^4\ge1/\chi$.  Hence bounded collision moment requires
$\Omega(N)$ mutually orthogonal Slater determinants.

For example, let
\begin{equation}
 \ket{\Psi_{\rm cat}}=
 \frac{\ket{f_1}+\ket{f_2}}{\sqrt2},
 \label{eq:smcat}
\end{equation}
where $f_1$ and $f_2$ fill complementary $2N$-dimensional orbital
subspaces $F_1$ and $F_2=F_1^\perp$. One-body cross terms vanish, but the collision moment
remains extensive, with
\begin{equation}
 \begin{aligned}
 \Gamma&=\frac{P_{F_1}+P_{F_2}}2=\frac{\id}{2},
 &\nu&=N,\\
 \A_2(\Psi_{\rm cat})&=\frac12+\frac1{2N+1},
 &\CN_N(\Psi_{\rm cat})&>\frac N2 .
 \end{aligned}
 \label{eq:smcatcollision}
\end{equation}
Explicitly, the diagonal doubled determinants contribute
$1/2$ to the singlet weight.  In either cross sector there is one
replica fermion per physical mode, and the direct and exchange singlet
matrix elements are
$\int_0^1t^{2N}dt$ and $\int_0^1(1-t)^{2N}dt$, with
$t=|g_{11}|^2$ uniform on $[0,1]$ for Haar $\SU(2)$.  The
exchange and fermionic-reordering signs cancel (each is an even power
of $-1$ for groups of $2N$ fermions), giving the remaining
$1/(2N+1)$ in Eq.~\eqref{eq:smcatcollision}.

The cat state also fails to anticoncentrate.  If
$\ket P=V^\dagger\ket x$ is Haar on $\Gr(2N,4N)$, then for fixed
$0<r<1$,
\begin{equation}
 \begin{aligned}
 p_x&=\frac12
 \left|\braket P{f_1}+\braket P{f_2}\right|^2,\\
 \E_VX_x^r
 &\le C_rD_N^r\sum_{k=1}^2
 \E_P|\braket P{f_k}|^{2r}
 =O_r\!\left(N^{-r(1-r)/2}\right),
 \end{aligned}
 \label{eq:smcatfractional}
\end{equation}
by the Slater-determinant law \eqref{eq:smfreelaw}.  Markov's inequality then gives
$\Pr_V[X_x\ge\alpha]\to0$ for every fixed $\alpha>0$.

\section{Higher replicas and the Porter--Thomas threshold}
\label{sm:higher}

In this section, we extend the two-copy sector reduction to a fixed number $\ell\ge2$ of replicas and identify the permutation saddles that determine the moment hierarchy. Their contributions distinguish collision anticoncentration from agreement with all Porter--Thomas moments. We also control fractional moments at asymptotically maximal density, where the limiting probability law is fixed by the integer moments.

\begin{proposition}[Flag reduction and Laplace asymptotic]\label{smprop:flag}
For the four-mode product family,
\begin{equation}
  \A_\ell(\Psi_{\bm\lambda})
  =
  \int_{\Flag_\ell}
  \prod_{i=1}^N\phi_{\lambda_i}^{(\ell)}(g)\,
  d\mu_{\Flag_\ell}(g),
  \qquad
  \Flag_\ell=\U(\ell)/\U(1)^\ell ,
  \label{eq:smflag}
\end{equation}
where
\begin{equation}
  \phi_\lambda^{(\ell)}(g)
  =
  \sum_{j=0}^{\ell}
  \lambda^j(1-\lambda)^{\ell-j}
  \sum_{\substack{A,B\subseteq[\ell]\\|A|=|B|=j}}
  |\det g_{B,A}|^4 .
\end{equation}
Its large-$\nu$ behavior is
\begin{equation}
  \A_\ell(\Psi_{\bm\lambda})
  =
  \ell!\,\GB(\ell+1)\,
  \nu^{-\binom{\ell}{2}}[1+o_\ell(1)]
  \qquad (\nu\to\infty).
  \label{eq:smlaplace}
\end{equation}
\end{proposition}

\begin{proof}
The projector onto replica $\det^n$, $n=2N$, is the character average
\begin{equation}
  \Pi_\ell
  =\int_{\U(\ell)}
  \overline{\det(g)^n}\,\pi_\ell(g)\,dg ,
\end{equation}
where $\pi_\ell$ is the replica-$\U(\ell)$ action on the $\ell$-fold
replicated Fock space.  For one replicated block, let $A\subseteq[\ell]$ label the
replicas in configuration $\ket{1100}$.  If $|A|=|B|=j$, then
\begin{equation}
  \matel{B}{\pi_\ell(g)}{A}
  =
  \det(g_{B,A})^2\det(g_{B^c,A^c})^2 .
  \label{eq:smminors}
\end{equation}
The reordering sign between replica-first and mode-first conventions
depends only on $j=|A|$. Since $\pi_\ell(g)$ preserves the total replica
occupation of every physical mode, a nonzero matrix element has
$|B|=|A|$, so the bra and ket signs cancel.
For unitary $g$, the complementary-minor identity turns this into
$\det(g)^2|\det(g_{B,A})|^4$.  Across $N$ blocks,
$\det(g)^{2N}$ cancels the character, and diagonal phases drop out,
leaving Eq.~\eqref{eq:smflag}.

The integral is evaluated by Laplace's method. The maxima of every
nonfree block symbol are the $\ell!$ Weyl points represented by
permutation matrices. At a permutation matrix $P$, use
$g=P e^{K(Z)}$, where $K_{a,b}=z_{a,b}$ and
$K_{ba}=-\overline{z}_{a,b}$ for $a<b$, $K_{aa}=0$, and
$\|Z\|^2=\sum_{a<b}|z_{a,b}|^2$.
The normalized invariant measure has
\begin{equation}
  d\mu_{\Flag_\ell}(Z)
  =
  \frac{\GB(\ell+1)}{\pi^{\binom{\ell}{2}}}
  [1+O_\ell(\|Z\|^2)]
  \prod_{a<b}d^2z_{a,b}.
  \label{eq:smflagmeasure}
\end{equation}
The normalization follows by induction on $\ell$ from the affine
Fubini--Study factor $(\ell-1)!/\pi^{\ell-1}$.

For the symbol, set
\begin{equation}
  d_j(g)=\binom{\ell}{j}
  -\sum_{\substack{|A|=|B|=j}}|\det g_{B,A}|^4 .
\end{equation}
Since $\wedgepow jg$ is unitary, $d_j\ge0$, and
\begin{equation}
  1-\phi_\lambda^{(\ell)}(g)
  =
  \lambda(1-\lambda)
  \sum_{j=1}^{\ell-1}
  \lambda^{j-1}(1-\lambda)^{\ell-j-1}d_j(g).
  \label{eq:smstifffactor}
\end{equation}
Near a Weyl point,
\begin{equation}
  \begin{aligned}
  d_j(Z)
  &=4\binom{\ell-2}{j-1}\|Z\|^2
    +O_\ell(\|Z\|^3),\\
  \phi_\lambda^{(\ell)}(Z)
  &=1-m\|Z\|^2+O_\ell(m\|Z\|^3),
  \qquad m=4\lambda(1-\lambda).
  \end{aligned}
  \label{eq:smhessian}
\end{equation}
Moreover $d_1=\ell-\sum_{a,b}|g_{a,b}|^4$ vanishes exactly on monomial
unitaries.  Its coefficient in Eq.~\eqref{eq:smstifffactor} is at
least $2^{-\ell}m$ for $0\le\lambda\le1/2$.
Compactness supplies constants $c_\ell,r_\ell>0$ such that
\begin{equation}
  1-\phi_\lambda^{(\ell)}(Z)
  \ge c_\ell m\|Z\|^2
  \quad(\|Z\|\le r_\ell),
  \qquad
  1-\phi_\lambda^{(\ell)}(g)
  \ge c_\ell m
\end{equation}
away from the $r_\ell$-neighborhood of the Weyl set.

Since $\phi_\lambda^{(\ell)}\ge(1-\lambda)^\ell\ge2^{-\ell}$,
$\log\phi_\lambda^{(\ell)}$ is uniformly bounded.
Equations~\eqref{eq:smhessian} and
$\sum_im_i^2\le\nu$ imply
\begin{equation}
  \sum_i\log\phi_{\lambda_i}^{(\ell)}(Z)
  =
  -\nu\|Z\|^2+O_\ell(\nu\|Z\|^3).
  \label{eq:smuniformlog}
\end{equation}
The complement of the Weyl neighborhoods contributes
$O(e^{-c_\ell\nu})$, and the annulus
$\nu^{-2/5}\le\|Z\|\le r_\ell$ is
$O(e^{-c_\ell\nu^{1/5}})$.  On $\|Z\|\le\nu^{-2/5}$,
Eq.~\eqref{eq:smflagmeasure} gives, at one Weyl point,
\begin{equation}
  \frac{\GB(\ell+1)}{\pi^{\binom{\ell}{2}}}
  \int_{\CC^{\binom{\ell}{2}}}
  e^{-\nu\|Z\|^2}\prod_{a<b}d^2z_{a,b}
  [1+o_\ell(1)]
  =
  \frac{\GB(\ell+1)}
       {\nu^{\binom{\ell}{2}}}[1+o_\ell(1)] .
  \label{eq:smgauss}
\end{equation}
Multiplying by the $\ell!$ Weyl points yields
Eq.~\eqref{eq:smlaplace}.
\end{proof}

Combining Eqs.~\eqref{eq:smRq} and \eqref{eq:smlaplace} gives
\begin{equation}
  \E_V[X_x^\ell]
  =
  \ell!\left(\frac{N}{\nu}\right)^{\binom{\ell}{2}}
  [1+o_\ell(1)].
  \label{eq:smPT}
\end{equation}
When $\nu/N\to\bar\gamma\in(0,1]$, the limit is
$\ell!\bar\gamma^{-\binom{\ell}{2}}$.  Hence $\gamma\to1$ implies the full
fixed-order Porter--Thomas hierarchy.  Conversely, its second moment
must tend to $2$. Equation~\eqref{eq:smbracket} first gives
$\nu\to\infty$, and Eq.~\eqref{eq:smsharp} then forces $\nu/N\to1$.
At $N=1$ and $\ell=3$, using
$\E_{\U(3)}|g_{a,b}|^4=1/6$ gives
\begin{eqnarray}
  \phi_\lambda^{(3)}(g)
  &=&1-\frac{m}{4}
    \left(3-\sum_{a,b=1}^3|g_{a,b}|^4\right),\\
  \A_3&=&1-\frac{3m}{8},\\
  \left.\E_V[X_x^3]\right|_{m=1}
  &=&\frac{108}{25}\frac58=\frac{27}{10}.
  \label{eq:sm2710}
\end{eqnarray}
At fixed $0<\gamma\le1$, the limiting integer moments coincide with
those of $EL_\gamma$, where $E\sim\operatorname{Exp}(1)$ is
independent of $\zeta\sim\mathcal N(0,1)$ and
\begin{equation}
  \begin{aligned}
  L_\gamma
  &:=\exp\!\left[
    \sqrt{\log(1/\gamma)}\,\zeta
    -\frac12\log(1/\gamma)\right],\\
  \E[(EL_\gamma)^\ell]
  &=\ell!\,\gamma^{-\ell(\ell-1)/2}.
  \end{aligned}
  \label{eq:smlognormal}
\end{equation}
For $0<\gamma<1$ the lognormal factor is moment
indeterminate~\cite{Heyde1963Property}, and this indeterminacy survives
multiplication by $E$: for any positive $L$ independent of $E$,
$\E[e^{-tEL}]=\E[(1+tL)^{-1}]$ for $t>0$, the Stieltjes transform of $L$.
Thus the integer moments identify the limit law only at $\gamma=1$.

\subsection{Fractional moments at asymptotically maximal density}
\label{sm:fracpt}

At $\gamma\to1$ the fractional moments follow from the integer ones.
Equation~\eqref{eq:smPT} sends every
integer moment to $\ell!$, the moments of $\mathrm{Exp}(1)$, which
satisfy Carleman's condition, so
$X_x\Rightarrow\mathrm{Exp}(1)$.  Since
$\sup_N\E[(X_x^r)^{1/r}]=\sup_N\E X_x=1$ and $1/r>1$, the family
$\{X_x^r\}$ is uniformly integrable.  Hence, for the four-mode product
family with $\nu\to\infty$ and $\nu/N\to1$,
\begin{equation}
  \E_V[X_x^r]\longrightarrow\Gamma(1+r),
  \qquad 0<r<1 .
  \label{eq:smfracpt}
\end{equation}

\section[Fractional moments]{Universal fractional-moment bounds and a Mellin asymptotic conjecture}
\label{sm:nogo}

Fractional moments detect whether typical rescaled probabilities vanish even when rare outcomes keep their mean fixed. We prove a universal bound by resolving the input into natural-orbital excitation sectors and estimating random Slater overlaps. We then separate exact padding identities and established distributional limits from the remaining Mellin asymptotic conjecture.
In this section $n$ particles occupy $2n$ modes and
$D=\binom{2n}{n}$; the main-text convention is $n=2N$.
We identify $V^\dagger\ket{x}$ with the Slater determinant
$\ket P$ of a Haar-random plane $P\in\Gr(n,2n)$. Thus
\begin{equation}
  \E_V[X_x^r]
  =D^r\E_P|\braket{P}{\Psi}|^{2r}.
  \label{eq:smgrassmannmoment}
\end{equation}

\subsection{Proof of Theorem~\ref{thm:noroute}}

The proof reduction is given in End Matter, Sec.~\ref{em:nogo}.
Here we establish its ingredients: the free Slater law, the bound on
the mean excitation number, and the uniform Jacobi estimate of
Lemma~\ref{smlem:dyadic}. Proposition~\ref{smprop:sparse} also treats
a general frozen Slater core.

\smallskip
\paragraph*{A frozen Slater core.}

\begin{proposition}[Frozen-core bound]\label{smprop:sparse}
Let $\ket{\Psi}=\ket f\wedge\ket\eta$, where $\ket f$ is an
$(n-q)$-particle Slater determinant and $\ket\eta$ is an arbitrary
normalized $q$-fermion state vector on its one-particle orthogonal
complement. For every $0<r<1$,
\begin{equation}
  D^r\E_P|\braket{P}{\Psi}|^{2r}
  \le C_r\left(\frac{q+1}{n}\right)^{r(1-r)/2}.
  \label{eq:smsparse}
\end{equation}
\end{proposition}

\begin{proof}
Put $\Delta=\|\iota_fP\|^2$, the squared norm of contracting the
Pl\"ucker state vector $\ket P$ against the frozen core $\ket f$.
Conditional on $\Delta>0$, the
normalized residual is a Haar-random $q$-plane
$\zeta\in\Gr(q,n+q)$, independent of $\Delta$, and
\begin{equation}
  |\braket{P}{f\wedge\eta}|^2
  =\Delta|\braket{\zeta}{\eta}|^2.
\end{equation}
Jensen's inequality for $t\mapsto t^r$ and Schur's lemma on the
irreducible module $\wedgepow{q}\CC^{n+q}$ give
\begin{equation}
  \E_P|\braket{P}{\Psi}|^{2r}
  \le \E\Delta^r
  \left(\E_\zeta|\braket{\zeta}{\eta}|^2\right)^r
  =\frac{\E\Delta^r}{\binom{n+q}{q}^{\,r}}.
\end{equation}
The Jacobi--Selberg moment of the contraction is
\begin{equation}
  \E\Delta^s
  =\prod_{j=0}^{n-q-1}
  \frac{\Gamma(j+1+q+s)\Gamma(n+j+1+q)}
       {\Gamma(j+1+q)\Gamma(n+j+1+q+s)}.
  \label{eq:smselberg}
\end{equation}
At $s=1$, this telescopes to
$\E\Delta=\binom{n+q}{q}/D$. Therefore the left-hand side of
Eq.~\eqref{eq:smsparse} is at most
$\E\Delta^r/(\E\Delta)^r$. Using, uniformly for integers
$\ell\ge1$,
\begin{equation}
 \log\frac{\Gamma(\ell+r)}{\Gamma(\ell)\ell^r}
 =-\frac{r(1-r)}{2\ell}+O_r(\ell^{-2}),
\end{equation}
in Eq.~\eqref{eq:smselberg}, the $O_r(\ell^{-2})$ terms sum to
$O_r(1)$, whereas the first-order terms give
\begin{equation}
 \sum_{j=0}^{n-q-1}
 \left(\frac1{j+1+q}-\frac1{n+j+1+q}\right)
 =\log\frac{n}{q+1}+O(1),
\end{equation}
and the claim follows.
\end{proof}

At $q=0$, Eq.~\eqref{eq:smselberg} gives, for every fixed $r>0$, the
free moment
\begin{equation}
 \begin{aligned}
  \E_Pp_x^r
  &=\prod_{j=0}^{n-1}
    \frac{\Gamma(j+1+r)\Gamma(n+j+1)}
         {\Gamma(j+1)\Gamma(n+j+1+r)},\\
  F_r(n)
  &:=D^r\E_Pp_x^r
    =\Theta_r(n^{-\beta}),
  \qquad \beta:=\frac{r(1-r)}2,\quad n\to\infty .
 \end{aligned}
  \label{eq:smfreelaw}
\end{equation}

\smallskip
\paragraph*{Natural-orbital reduction.}
Use the natural-orbital split $F\oplus G$, sector weights $w_j$,
and mean excitation number $\mu$ defined in End Matter,
Sec.~\ref{em:nogo}.
If $\theta_1\ge\cdots\ge\theta_{2n}$ are the natural occupations,
then either $\theta_{n+1}\le1/2$ or $\theta_n>1/2$.  Applying
$t(1-t)\le1-t$ on $F$, $t(1-t)\le t$ on $G$, and
$\min(t,1-t)\le2t(1-t)$ to the lower or upper half, respectively, gives
\begin{equation}
  \begin{aligned}
  \mu:=\sum_jjw_j
  &=\sum_{i>n}\theta_i
   =\sum_{i\le n}(1-\theta_i),\\
  \nu&\le
  \sum_{i\le n}(1-\theta_i)+\sum_{i>n}\theta_i=2\mu,\\
  \mu&\le\sum_i\min(\theta_i,1-\theta_i)
  \le2\sum_i\theta_i(1-\theta_i)=2\nu .
  \end{aligned}
  \label{eq:smmunu}
\end{equation}

Let $y_1,\ldots,y_n$ be the squared principal sines of $P$ from
$F$. Write $Y=(y_1,\ldots,y_n)$ and
$\Delta(Y)=\prod_{a<b}(y_a-y_b)$ for the Vandermonde determinant.
Their flat Jacobi unitary law is $d\mathbb P_0\propto\Delta(Y)^2dY$. With
$x_i=y_i/(1-y_i)$, the normalized sector likelihoods obey
\begin{equation}
 \begin{aligned}
 Z_j(Y)
 &:=\frac{D}{\binom nj^2}
 \prod_i(1-y_i)e_j(x_1,\ldots,x_n),\\
 \E_0Z_j&=1,
 &D\,p_x(\Tw_H\rho)&=\sum_{j=0}^n w_jZ_j .
 \end{aligned}
  \label{eq:smsectorlikelihood}
\end{equation}
Here $e_j$ is the $j$th elementary symmetric polynomial.
\begin{lemma}[Dyadic Jacobi bound]\label{smlem:dyadic}
For $0<r<1$ and $1\le L\le n/8$,
\begin{equation}
 \E_0\left[\sup_{L\le j<2L}Z_j^r\right]
 \le C_r\left(\frac Ln\right)^{\beta}.
 \label{eq:smdyadic}
\end{equation}
Here $\beta=r(1-r)/2$.
\end{lemma}

\begin{proof}
All constants in this proof depend only on fixed $r$, uniformly in
$n,L,j$, and $s$ over the stated ranges. The exponent
$\beta=r-r/2-r^2/2$ combines the prefactor $j^r$, the
bulk-spreading factor $J^{-r/2}$, and the tilted partition-function factor
$J^{-r^2/2}$.

\medskip\noindent\emph{Free factor and fugacity.}
Tilt $\mathbb P_0$ by $\prod_i(1-y_i)^r$, and write $\mathbb P_r$
and $\E_r$ for the resulting law and expectation. Since
$\prod_i(1-y_i)$ is the squared overlap with the Slater determinant
filling $F$, the tilt normalization is $F_r(n)/D^r$. With
$h_j=e_j(x_1,\ldots,x_n)/\binom nj^2$, the tilt identity and
Eq.~\eqref{eq:smfreelaw} give
\begin{equation}
 \begin{aligned}
 \E_0\left[\sup_{L\le j<2L}Z_j^r\right]
 &=F_r(n)\,\E_r\left[\sup_{L\le j<2L}h_j^r\right],\\
 F_r(n)&\le C_rn^{-\beta}.
 \end{aligned}
 \label{eq:smflattilt}
\end{equation}
For fixed $Y$, $e_j(x_1,\ldots,x_n)$ is the coefficient of $s^{2j}$
in $\prod_i(1+s^2x_i)$. Define
\begin{equation}
 \begin{aligned}
 p_s(y)&:=\frac{s^2y}{1-y+s^2y},
 &J(s)&:=\frac{ns}{1+s},\\
 M_s&:=\sum_i p_s(y_i),
 &V_s&:=\sum_i p_s(y_i)\bigl(1-p_s(y_i)\bigr),\\
 \overline A_s&:=
 \frac{\prod_i(1+s^2x_i)}{(1+s)^{2n}}.
 \end{aligned}
\end{equation}
Let $\pi_s$ be the Poisson--binomial law with parameters $p_s(y_i)$.
Set $s_j=j/(n-j)$, so that $J(s_j)=j$. Coefficient extraction,
Stirling's bound, and Fourier inversion give, with an absolute constant
uniform in $s$, the parameters $Y$, and $j\le n/4$,
\begin{equation}
 \begin{aligned}
 h_j
 &=\frac{s_j^{-2j}(1+s_j)^{2n}}{\binom nj^2}\,
   \overline A_{s_j}\pi_{s_j}(j),
 &\qquad
 \frac{s_j^{-2j}(1+s_j)^{2n}}{\binom nj^2}
 &\le Cj,\\
 \sup_k\pi_s(k)
 &\le\frac{C}{\sqrt{1+V_s}},
 &\qquad
 B_s&:=\overline A_s^r(1+V_s)^{-r/2},\\
 h_j^r&\le C_rj^rB_{s_j}.
 \end{aligned}
 \label{eq:smhjB}
\end{equation}

\medskip\noindent\emph{Tilted Jacobi ensemble.}
Multiplication by $\overline A_s^r$ changes $\mathbb P_r$ to the
measure $\mathbb Q_s$ defined by
\begin{equation}
 \begin{aligned}
 d\mathbb Q_s
 &=
 \frac{\Delta(Y)^2\prod_i(1-y_i+s^2y_i)^r\,dY}
      {\mathcal Z_n(s)},
 &\qquad
 \mathcal Z_n(s)
 &:=
 \int_{[0,1]^n}\Delta(Y)^2
 \prod_i(1-y_i+s^2y_i)^r\,dY,\\
 H(s)
 &:=\E_r\overline A_s^r
 =(1+s)^{-2rn}\frac{\mathcal Z_n(s)}{\mathcal Z_n(0)}.
 \end{aligned}
\end{equation}
Write $\E_s$ for expectation under $\mathbb Q_s$. Then
$\E_r[\overline A_s^r\Phi]=H(s)\E_s\Phi$.
Put $d_s(z):=s^2+(1-s^2)z$. The inverse of $z=p_s(y)$ satisfies
\begin{equation}
 y=\frac{z}{d_s(z)},\qquad
 dy=\frac{s^2}{d_s(z)^2}\,dz,\qquad
 1-y+s^2y=\frac{s^2}{d_s(z)},
\end{equation}
and $y_i-y_j=s^2(z_i-z_j)/[d_s(z_i)d_s(z_j)]$. Consequently,
\begin{equation}
 d\mathbb Q_s(z)\propto
 \Delta(z)^2\prod_i[d_s(z_i)]^{-(2n+r)}dz_i,
 \qquad M_s=\sum_i z_i,\quad V_s=\sum_i z_i(1-z_i).
 \label{eq:smmobiusope}
\end{equation}
This is a projection orthogonal-polynomial ensemble. Let
$\phi_k$ be the degree-$k$ orthonormal polynomial in
$L^2([0,1],d_s(z)^{-(2n+r)}dz)$, and let $K_s$ be the associated
rank-$n$ projection kernel. The linear-statistic variance formula and the
three-term recurrence then give~\cite{Forrester2010LogGases}
\begin{equation}
 \Var_sM_s
 =\lVert(\id-K_s)zK_s\rVert_{\rm HS}^2
 =|\langle\phi_n,z\phi_{n-1}\rangle|^2\le1.
 \label{eq:smrankonevar}
\end{equation}
\medskip\noindent\emph{Ward estimates.}
The loop equations generated by $z\mapsto z+t\,v(z)$ control the gap
between the tilted mean and its reference value. For $0<s<1$, write
\begin{equation}
 \delta_s:=J(s)-\E_sM_s,
 \qquad
 \chi_s:=s^{-2}-1.
\end{equation}
The boundary term vanishes because $v(0)=v(1)=0$, while the
Vandermonde factor makes the symmetrized pair term regular at
collisions. Integration by parts in Eq.~\eqref{eq:smmobiusope} with
$v(z)=z(1-z)[s^2+(1-s^2)z]$, followed by division by $s^2$ and subtraction of
$\chi_sJ^2+2nJ=n^2$, gives
\begin{equation}
 \begin{aligned}
 (1-s^2)\E_s[M_s^2]+2ns^2\E_sM_s
 +r(1-s^2)\E_sV_s
 &=s^2n^2,\\
 \delta_s\bigl[2n+\chi_s(J(s)+\E_sM_s)\bigr]
 &=\chi_s\bigl(\Var_sM_s+r\E_sV_s\bigr).
 \end{aligned}
 \label{eq:smwardgap}
\end{equation}
The second line implies $\E_sM_s\le J(s)$. Together with
Eq.~\eqref{eq:smrankonevar} and $V_s\le M_s$, it also gives
\begin{equation}
 0\le\delta_s\le r+J(s)^{-1}.
 \label{eq:smdeltaupper}
\end{equation}

Let primes denote differentiation with respect to $\log s$, with
the original variables $Y$ held fixed in the samplewise derivatives.
Differentiation under the integral is justified locally for $s>0$,
since $1-y+s^2y$ is bounded away from zero on compact $s$-intervals. Then
\begin{equation}
 \begin{aligned}
 M_s'&=2V_s,
 &J'(s)&=\frac{J(s)}{1+s},\\
 \left(\log\frac{d\mathbb Q_s}{dY}\right)'
 &=2r\bigl(M_s-\E_sM_s\bigr),
 &(\log H)'&=-2r\delta_s,\\
 \bigl(\E_sM_s\bigr)'
 &=2\E_sV_s+2r\Var_sM_s\\
 &\le\frac2r\bigl(\Var_sM_s+r\E_sV_s\bigr)\\
 &=\frac2r\,\chi_s^{-1}\delta_s
   \bigl[2n+\chi_s(J(s)+\E_sM_s)\bigr]\\
 &\le\frac{4J(s)}{r(1-s)}\,\delta_s.
 \end{aligned}
\end{equation}
Here the first inequality uses $0<r<1$, the middle equality uses the
second line of Eq.~\eqref{eq:smwardgap}, and the last uses
$\E_sM_s\le J(s)$ together with
\begin{equation}
 2n+\chi_s(J(s)+\E_sM_s)\le\frac{2n}{s},
 \qquad
 \chi_s^{-1}\frac{2n}{s}=\frac{2J(s)}{1-s}.
\end{equation}
Since $\delta_s'=J'(s)-(\E_sM_s)'$,
\begin{equation}
 \delta_s'
 \ge\frac{J(s)}{1+s}
 -\frac{4J(s)}{r(1-s)}\,\delta_s.
 \label{eq:smgapODE}
\end{equation}
It follows that
\begin{equation}
 \delta_s\ge
 \frac{r(1-s)}{4(1+s)}-\frac1{J(s)},
 \qquad 0<s\le\frac13.
 \label{eq:smdeltalower}
\end{equation}
Indeed, the difference between the two sides is positive for small
$s$, since the right-hand side tends to $-\infty$. At a first contact,
nonnegativity of $\delta_s$ gives
$J(s)^{-1}\le r(1-s)/[4(1+s)]$. Equation~\eqref{eq:smgapODE} then gives
$\delta_s'\ge4/[r(1-s)]>r/4$. At the same contact point, the $\log s$
derivative of the right-hand side of Eq.~\eqref{eq:smdeltalower} obeys
\begin{equation}
 -\frac{rs}{2(1+s)^2}+\frac1{J(s)(1+s)}\le\frac r4.
\end{equation}
This contradiction establishes Eq.~\eqref{eq:smdeltalower}.

Since $H(0)=1$ and $(\log H)'=-2r\delta_s\le0$, we have
$H(1/n)\le1$. Using this bound, Eq.~\eqref{eq:smdeltalower}, and
\begin{equation}
 \begin{aligned}
 \int_{1/n}^s\frac{1-\sigma}{1+\sigma}\frac{d\sigma}{\sigma}
 &=\log(ns)-2\log\frac{1+s}{1+1/n},\\
 \int_{1/n}^s\frac{d\sigma}{J(\sigma)\sigma}&=O(1),
 \end{aligned}
\end{equation}
integration of $(\log H)'=-2r\delta_s$ gives
\begin{equation}
 H(s)=\E_r\overline A_s^r\le C_rJ(s)^{-r^2/2},
 \qquad 1\le J(s)\le n/4.
 \label{eq:smsharponepoint}
\end{equation}
\medskip\noindent\emph{Bulk spreading.}
The
second identity in Eq.~\eqref{eq:smwardgap} gives
\begin{equation}
 r\E_sV_s
 =\delta_s\chi_s^{-1}
 \bigl[2n+\chi_s(J(s)+\E_sM_s)\bigr]
 -\Var_sM_s
 \ge \delta_sJ(s)-1,
\end{equation}
where we used
$\chi_s^{-1}[2n+\chi_s(J(s)+\E_sM_s)]\ge J(s)$
and Eq.~\eqref{eq:smrankonevar}. Since $J(s)\le n/4$ implies
$s\le1/3$, Eqs.~\eqref{eq:smdeltalower} and $J(s)\ge32/r$ yield
\begin{equation}
 \delta_s\ge\frac r8-\frac1{J(s)}\ge\frac{3r}{32},
 \qquad
 \E_sV_s\ge\frac{3J(s)}{32}-\frac1r\ge\frac{J(s)}{16}.
\end{equation}
With $h(z)=1-e^{-z(1-z)}\ge cz(1-z)$, the determinantal Laplace
functional gives~\cite{Forrester2010LogGases}
\begin{equation}
 \E_se^{-V_s}
 =\det(\id-\sqrt hK_s\sqrt h)
 \le e^{-\tr(hK_s)}
 \le\exp[-c\,\E_sV_s]
 \le e^{-c_rJ(s)},
 \qquad J(s)\ge32/r.
\end{equation}
Markov's inequality gives
$\mathbb Q_s(V_s\le c_0J(s))\le e^{-c_rJ(s)}$ for some $c_0>0$
and $J(s)\ge32/r$. Splitting on this event, and enlarging the constant
over the remaining compact range $1\le J(s)<32/r$, gives
\begin{equation}
 \E_s(1+V_s)^{-r/2}\le C_rJ(s)^{-r/2},
 \qquad
 \E_s(1+V_s)^{-r}\le C_rJ(s)^{-r},
 \qquad 1\le J(s)\le n/4.
 \label{eq:smnegmom}
\end{equation}

\medskip\noindent\emph{Dyadic maximum.}
The interval $[\log s_L,\log s_{2L}]$ has length bounded above and
below by absolute constants, and $L\le J(s)\le2L$ throughout. Along it,
\begin{equation}
 \begin{aligned}
 \partial_{\log s}\log\overline A_s
 &=2(M_s-J(s)),
 &\qquad
 \left|\partial_{\log s}\log(1+V_s)\right|
 &\le2,\\
 \left|\partial_{\log s}B_s\right|
 &\le C_rB_s\bigl(1+|M_s-J(s)|\bigr).
 \end{aligned}
\end{equation}
Since $M_s-J(s)=(M_s-\E_sM_s)-\delta_s$, Cauchy--Schwarz and
Eqs.~\eqref{eq:smdeltaupper}, \eqref{eq:smrankonevar}, and
\eqref{eq:smnegmom} give
\begin{equation}
 \begin{aligned}
 &\E_s\!\left[
 (1+V_s)^{-r/2}\bigl(1+|M_s-J(s)|\bigr)\right]\\
 &\quad\le
 (1+\delta_s)\E_s(1+V_s)^{-r/2}
 +(\Var_sM_s)^{1/2}
  \bigl[\E_s(1+V_s)^{-r}\bigr]^{1/2}\\
 &\quad\le C_rJ(s)^{-r/2}.
 \end{aligned}
 \label{eq:smderivativecontrol}
\end{equation}
The change-of-measure identity
$\E_r[\overline A_s^r\Phi]=H(s)\E_s\Phi$, together with
Eqs.~\eqref{eq:smsharponepoint} and
\eqref{eq:smderivativecontrol}, gives
\begin{equation}
 \begin{aligned}
 \E_rB_s
 &\le C_rJ(s)^{-(r+r^2)/2},
 &\qquad
 \E_r\left|\partial_{\log s}B_s\right|
 &\le C_rJ(s)^{-(r+r^2)/2},\\
 \E_r\sup_{s_L\le s\le s_{2L}}B_s
 &\le C_rL^{-(r+r^2)/2},
 \end{aligned}
\end{equation}
where the last line follows from the samplewise fundamental theorem of
calculus and Tonelli, since $B_s$ is almost surely $C^1$ on this interval
and $\mathbb P_r$ is independent of $s$.
Equation~\eqref{eq:smhjB} therefore yields
\begin{equation}
 \E_r\sup_{L\le j<2L}h_j^r
 \le C_rL^r
 \E_r\sup_{s_L\le s\le s_{2L}}B_s
 \le C_rL^{\,r-(r+r^2)/2}
 =C_rL^{\beta}.
\end{equation}
Substituting this into Eq.~\eqref{eq:smflattilt} gives
Eq.~\eqref{eq:smdyadic}.
\end{proof}

Together with the Slater law~\eqref{eq:smfreelaw} and the
occupation bound~\eqref{eq:smmunu}, Lemma~\ref{smlem:dyadic}
supplies the estimates used in End Matter, Sec.~\ref{em:nogo}.
The summation in Eq.~\eqref{eq:emnogoreduction} then completes
the proof of Theorem~\ref{thm:noroute}.

\subsection{The Mellin asymptotic conjecture}

The comparison law $EL_\gamma$ in Eq.~\eqref{eq:smlognormal} has
Mellin transform
$\Gamma(1+r)\gamma^{r(1-r)/2}$.  This suggests the following
fractional-moment asymptotic.

\begin{conjecture}[Lognormal Mellin asymptotic]\label{smconj:mellin}
For every fixed $0<r<1$ and every four-mode product profile,
\begin{equation}
  \E_V[X_x^r]
  =\Gamma(1+r)
   \left(\frac{\nu}{N}\right)^{r(1-r)/2}[1+o(1)]
   \qquad (\nu\to\infty).
  \label{eq:smmellin}
\end{equation}
\end{conjecture}

Theorem~\ref{thm:noroute} proves a profile-independent upper bound with
the same exponent, but not the equality, since the natural-orbital twirl
discards coherent interference. Direct amplitudes for uniform profiles
at $N=12,14,\ldots,22$ are consistent with both the exponent and the prefactor
(Fig.~\ref{smfig:mellin}). The maximally concentrated, free-padded profiles are treated exactly below.

\begin{figure}[t]
\centering
\includegraphics[width=1\linewidth]{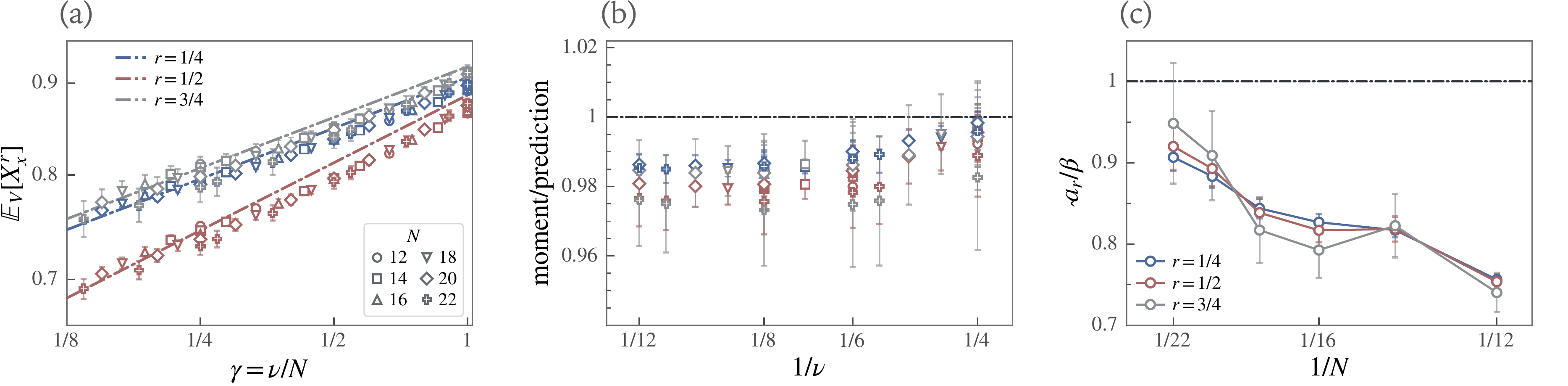}
\caption{\textbf{Finite-size tests of fractional-moment scaling.}
Direct amplitude evaluations use $6.6\times10^5$ independent Haar
draws at $N=12,14,\ldots,22$.
(a) Fractional moments for uniform profiles at
$r=1/4,1/2,3/4$, compared with
$\Gamma(1+r)\gamma^{r(1-r)/2}$ (dash-dotted).
(b) Ratios to this prediction over the fitting window
$4\le\nu\le0.6N$. The horizontal guide is one.
(c) Effective log--log slopes, divided by
$\beta=r(1-r)/2$, from the same window. Thin bars in (a),(b) are
pointwise $95\%$ batch intervals. Bars in (c) show Gaussian-equivalent
scales obtained from pointwise $95\%$ joint batch-bootstrap intervals;
fitting and resampling procedures are described in Sec.~\ref{sm:numerics}.
The finite-size trend is consistent with
Conjecture~\ref{smconj:mellin}. The maximally concentrated profiles are fixed
analytically by Corollary~\ref{smcor:padding}.}
\label{smfig:mellin}
\end{figure}

\subsection{Exact free padding}
\label{sm:padding}

Adjoining frozen occupied and empty modes separates the fractional moments into an active-state factor and an exactly computable Slater factor.
For a normalized state $\Psi\in\wedgepow{2N}\CC^{4N}$, define
\begin{equation}
 \mathsf M_r^{(N)}(\Psi)
 :=D_N^r\E_{P\sim\Gr(2N,4N)}
 |\langle P|\Psi\rangle|^{2r}
 =\E_V[X_x^r],
 \label{eq:smpaddef}
\end{equation}
so $\mathsf M_r^{(N)}(\text{Slater})=F_r(2N)$ by
Eq.~\eqref{eq:smfreelaw}. For
$\eta\in\wedgepow{2K}\CC^{4K}$, let $f_O\wedge\eta$ denote the state 
obtained by adjoining $2(N-K)$ frozen occupied and $2(N-K)$ frozen
empty modes, equivalently $N-K$ free ($m_i=0$) blocks.

\begin{theorem}[Exact free-padding factorization]\label{smthm:padding}
For every $r>0$, $N\ge K\ge1$, and a normalized state 
$\eta\in\wedgepow{2K}\CC^{4K}$, let $f_O$ be the Slater determinant
filling the $2(N-K)$ frozen occupied modes. Then
\begin{equation}
 \mathsf M_r^{(N)}(f_O\wedge\eta)
 =\frac{F_r(2N)}{F_r(2K)}\mathsf M_r^{(K)}(\eta).
 \label{eq:smpadratio}
\end{equation}
\end{theorem}

\begin{proof}
Write $\CC^{4N}=O\oplus A\oplus E$, with
$\dim O=\dim E=2(N-K)$ and $\dim A=4K$.  If $p(P)$ is a normalized
Pl\"ucker state  representing $P$, define
\begin{equation}
  \operatorname{proj}_{\wedgepow{2K}A}
  \!\left(\iota_{f_O}p(P)\right)
  =\varrho(P)\ket{Q(P)},\qquad Q(P)\in\Gr(2K,A).
\end{equation}
Exterior contraction gives
\begin{equation}
  |\langle P|f_O\wedge\eta\rangle|^2
  =\varrho(P)^2|\langle Q(P)|\eta\rangle|^2 .
\end{equation}
Under $\id_O\oplus u\oplus\id_E$, the radius $\varrho$ is fixed and
$Q\mapsto uQ$.  Haar invariance makes $Q$ Haar on
$\Gr(2K,A)$ and independent of $\varrho$.  Taking an arbitrary $\eta$ and
then a Slater determinant gives, respectively,
\begin{equation}
  \begin{aligned}
  \frac{\mathsf M_r^{(N)}(f_O\wedge\eta)}{D_N^r}
  &=\E\varrho^{2r}\,
    \frac{\mathsf M_r^{(K)}(\eta)}{D_K^r},\\
  \frac{F_r(2N)}{D_N^r}
  &=\E\varrho^{2r}\,\frac{F_r(2K)}{D_K^r}.
  \end{aligned}
\end{equation}
The ratio of these two identities is Eq.~\eqref{eq:smpadratio}.
\end{proof}

The Gamma product in Eq.~\eqref{eq:smfreelaw}, at $n=2N$, and the
standard Barnes-$G$ expansion give
\begin{equation}
 \begin{aligned}
 F_r(2N)
 &=\frac{N^{-\beta}}{\GB(1+r)}
  [1+O_r(N^{-1})],\\
 \frac{F_r(2N)}{F_r(2K)}
 &=\left(\frac KN\right)^{\beta}
  [1+O_r(K^{-1})]
 \end{aligned}
 \label{eq:smpadfreeasymp}
\end{equation}
for fixed $0<r<1$ and $N\ge K\to\infty$.

\begin{corollary}[Concentrated Mellin law]\label{smcor:padding}
Let $K=\lceil\nu\rceil\to\infty$, take
$\lfloor\nu\rfloor$ blocks with $m_i=1$ and at most one partial block,
and pad them by $N-K$ Slater blocks. Then, for $0<r<1$,
\begin{equation}
 \E_V[X_x^r]
 =\Gamma(1+r)
  \left(\frac{\nu}{N}\right)^{r(1-r)/2}[1+o(1)].
 \label{eq:smpadlaw}
\end{equation}
\end{corollary}

\begin{proof}
The active $K$-block state  $\eta$ has $\nu/K\to1$.
Theorem~\ref{thm:PT} and Eq.~\eqref{eq:smfracpt} therefore give
$\mathsf M_r^{(K)}(\eta)\to\Gamma(1+r)$. Apply
Theorem~\ref{smthm:padding}, Eq.~\eqref{eq:smpadfreeasymp}, and
$K=\nu[1+o(1)]$.
\end{proof}

\subsection{Distributional limits and remaining mathematical questions}
\label{sm:jeopen}

Fractional moments can resolve the distributional ambiguity left by the integer hierarchy. The following observation strengthens the concentrated-profile result without assuming Conjecture~\ref{smconj:mellin}.

\begin{lemma}[Fractional-moment determination]
\label{smlem:jefracunique}
Let $Y_N\ge0$ satisfy $\sup_N\E Y_N<\infty$, and let $Y\ge0$ have finite mean.
If $\E Y_N^r\to\E Y^r$ for every $r$ in a nonempty open interval $I\subset(0,1)$, then $Y_N\Rightarrow Y$.
\end{lemma}

\begin{proof}
The mean bound gives tightness, and for each $r\in I$ it gives uniform integrability of $Y_N^r$ by the bound on its $1/r$ moment.
Every subsequential weak limit $Z$ therefore satisfies $\E Z^r=\E Y^r$ on $I$.
Choose $r_0$ in the interior of $I$ and, for $W=Z,Y$, define the finite measure
$\mu_W(B)=\E[W^{r_0}\mathbf1_{\{W>0,\,\log W\in B\}}]$ on $\mathbb R$.
Its bilateral Laplace transform is $\E W^{r_0+s}$ for $s$ in a neighborhood of zero.
Equality of these transforms implies equality of the measures, by uniqueness of a Laplace transform finite near zero.
Undoing the strictly positive weight $e^{r_0 x}$ identifies the laws on $(0,\infty)$; their remaining mass at zero then agrees as well.
Thus every subsequential limit has the law of $Y$, proving convergence.
\end{proof}

\begin{proposition}[Distributional limit for concentrated profiles]
\label{smprop:jeconcentratedlaw}
For the concentrated profiles of Corollary~\ref{smcor:padding}, if
$\nu/N\to\bar\gamma\in(0,1]$, then
\begin{equation}
 X_x\Rightarrow E L_{\bar\gamma},
 \label{eq:jeconcentratedlaw}
\end{equation}
with the independent exponential and mean-one lognormal factors defined in Eq.~\eqref{eq:smlognormal}.
More generally, Conjecture~\ref{smconj:mellin}, if established for all $0<r<1$, implies the same limit for every four-mode profile with this limiting density.
\end{proposition}

\begin{proof}
Corollary~\ref{smcor:padding} gives convergence of every fractional moment to
$\Gamma(1+r)\bar\gamma^{r(1-r)/2}=\E[(EL_{\bar\gamma})^r]$.
Both $X_x$ and $EL_{\bar\gamma}$ have mean one, so Lemma~\ref{smlem:jefracunique} applies.
The same argument uses Conjecture~\ref{smconj:mellin} for general profiles.
\end{proof}

For general weak-block profiles, the missing step is therefore a fractional-moment or distributional estimate; the integer moments alone cannot select this law.
One possible route is to control the Laplace transform, since for $0<r<1$
\begin{equation}
 \E X_x^r=\frac{r}{\Gamma(1-r)}
 \int_0^\infty\bigl(1-\E e^{-tX_x}\bigr)t^{-r-1}\,dt,
 \qquad
 \E e^{-tEL_\gamma}=\E(1+tL_\gamma)^{-1}.
 \label{eq:jelaplaceroute}
\end{equation}
The first identity follows from the corresponding scalar integral and Tonelli's theorem. Establishing convergence of these Laplace transforms would select the candidate law, while the mean bound would supply the required fractional-moment uniform integrability.

Beyond product states, the replica reduction gives an exact target for a sufficient condition on correlations:
for every fixed $\ell\ge2$, the Porter--Thomas moment hierarchy is equivalent to
\begin{equation}
 \A_\ell(\Psi_N)=\ell!\,\GB(\ell+1)
 N^{-\binom{\ell}{2}}[1+o_\ell(1)].
 \label{eq:jecorrelationscriterion}
\end{equation}
This follows directly from Eq.~\eqref{eq:smRq}.
The open problem is to derive these sector-weight estimates from physically useful assumptions on the many-body correlations. Extensive defect alone cannot do so, as the cat example shows.

A separate computational question is whether the average-case reduction in Conjecture~\ref{smconj:achard} can be established for relative-error estimates on the specified finite ensemble.
The exact worst-case embedding of Sec.~\ref{sm:sharpP} does not supply such a reduction, so the approximate-sampling conclusion remains conditional even where the limiting probability law is known.

\section{Preparation cost}
\label{sm:preparation}

The output obstruction becomes a preparation bound once the defect is compared with a Gaussian resource monotone. We bound the contributions of non-Gaussian resource modes and parity-preserving local gates, and then quantify the stability of the bound when the prepared pure state only approximates the target.
For a pure fixed-parity state, let $G_\Psi$ be the Gaussian state
with the same full Majorana covariance matrix and set
\begin{equation}
 M(\Psi):=S(G_\Psi).
\end{equation}
If $r_a$ are the Williamson values of that covariance, then
$M(\Psi)=\sum_a h_2((1+r_a)/2)$. It is additive and strongly monotone
under pure-state Gaussian protocols~\cite{Turkeshi2026Williamson}. For a
fixed-number state with natural occupations $\theta_a$, anomalous
correlators vanish and $r_a=|1-2\theta_a|$. Hence
\begin{equation}
 M(\Psi)=\sum_a h_2(\theta_a)
 \ge4\sum_a\theta_a(1-\theta_a)=4\nu(\Psi),
 \label{eq:MboundSM}
\end{equation}
where $h_2$ is the binary entropy in bits. Here
$4\nu=\sum_a(1-r_a^2)$ is the fixed-number restriction of the quadratic
full-covariance antiflatness~\cite{Leone2026FermionicEntropy,
Turkeshi2026Williamson}. Equality in Eq.~\eqref{eq:MboundSM} holds
precisely when every $\theta_a\in\{0,1/2,1\}$.

\begin{lemma}[Entropy change under a local gate]
\label{smlem:localgate}
Let \(U_A\) be a parity-preserving unitary supported on \(q\)
fermionic modes.  For every pure fixed-parity state vector \(\ket{\Psi}\),
\begin{equation}
 \big|M(U_A\Psi)-M(\Psi)\big|\le2q .
 \label{eq:smlocalgate}
\end{equation}
The coefficient \(2\) is optimal, already for \(q=2\).
\end{lemma}

\begin{proof}
Let \(G\) and \(G'\) be the Gaussian states with the covariances of
\(\Psi\) and \(U_A\Psi\), respectively, and write \(B=A^c\).
Parity preservation leaves the \(B\times B\) covariance block
unchanged, hence \(G'_B=G_B\).  Therefore
\begin{equation}
 M(U_A\Psi)-M(\Psi)=S(A|B)_{G'}-S(A|B)_G .
\end{equation}
Since \(A\) contains \(q\) modes, \(-q\le S(A|B)\le q\), proving the
bound.  Sharpness is witnessed by the ordinary qubit
\(\mathrm{SWAP}_{23}\), which maps the Gaussian state vector
\(\ket{\Phi^+}_{12}\ket{\Phi^+}_{34}\) to
\(\ket{\Phi^+}_{13}\ket{\Phi^+}_{24}\), whose covariance vanishes,
and raises \(M\) from \(0\) to \(4\).
\end{proof}

\begin{proposition}[Preparation with resource states and local gates]
\label{smprop:hybridcost}
Consider a deterministic exact preparation of a pure fixed-number
target state \(\Psi\). The protocol may start from pure Gaussian states, consume a possibly joint pure fixed-parity resource state on \(q_0\) modes, apply parity-preserving non-Gaussian gates supported on \(q_t\) modes, \(t\ge1\), and
use arbitrary pure-state Gaussian operations between them.  Then
\begin{equation}
 q_0+2\sum_{t\ge1}q_t\ge M(\Psi)\ge4\nu(\Psi).
 \label{eq:costSM}
\end{equation}
\end{proposition}

\begin{proof}
The resource contributes at most \(q_0\) bits of \(M\), Gaussian
segments do not increase its branch average, and gate \(t\) changes it
by at most \(2q_t\) by Lemma~\ref{smlem:localgate}.  Iterating these
bounds proves Eq.~\eqref{eq:costSM}.  A pure catalyst returned
unchanged and in product with the output cancels by additivity.
\end{proof}

For \(R\) four-mode resources and \(n_{\rm SWAP}\) ordinary qubit
SWAPs, possibly used together,
\begin{equation}
 R+n_{\rm SWAP}\ge\frac{M(\Psi)}4\ge\nu(\Psi).
 \label{eq:smswapcombined}
\end{equation}
The SWAP-only specialization recovers
Ref.~\cite{Tarabunga2026Computable}; Eq.~\eqref{eq:costSM} extends it
to arbitrary resource sizes and parity-preserving mode-local gates.
For outcome-dependent gates, Eq.~\eqref{eq:costSM} holds with
\(\sum_t q_t\) replaced by \(\E_b[\sum_t q_t(b)]\), where \(b\)
is the Gaussian measurement record and \(q_t(b)=0\) after termination.

\begin{lemma}[Trace-distance stability of the defect]
\label{smlem:lipschitz}
Density operators \(\rho_\Phi,\rho_\Psi\) of \(2N\) fermions on \(4N\) modes with
\(\tfrac12\|\rho_\Phi-\rho_\Psi\|_1\le\epsilon\) obey
\(|\nu(\rho_\Phi)-\nu(\rho_\Psi)|\le4N\epsilon\).
For pure inputs, \(\rho_\Phi=\ket\Phi\bra\Phi\) and
\(\rho_\Psi=\ket\Psi\bra\Psi\).
\end{lemma}

\begin{proof}
Let \(\Gamma_\Phi,\Gamma_\Psi\) be the one-body correlation matrices of \(\rho_\Phi,\rho_\Psi\), respectively. Write \(\Delta\Gamma:=\Gamma_\Phi-\Gamma_\Psi\) and
\(\Delta\rho:=\rho_\Phi-\rho_\Psi\), so \(\|\Delta\rho\|_1\le2\epsilon\) and
\(\tr\Delta\Gamma=0\) at fixed particle number. For a Hermitian
single-particle contraction \(W\) with \(\|W\|_\infty\le1\), the second-quantized
operator \(d\Gamma(W)=\sum_{a,b}W_{ba}c_a^\dagger c_b\) has spectrum in
\([-2N,2N]\) on the \(2N\)-particle sector, each eigenvalue being a
sum of \(2N\) eigenvalues of \(W\).  Hence
\begin{equation}
 \|\Delta\Gamma\|_1
 =\sup_{\substack{W=W^\dagger\\\|W\|_\infty\le1}}
   \big|\tr(W\Delta\Gamma)\big|
 =\sup_{\substack{W=W^\dagger\\\|W\|_\infty\le1}}
   \big|\tr(d\Gamma(W)\,\Delta\rho)\big|
 \le2N\|\Delta\rho\|_1\le4N\epsilon .
\end{equation}
Since \(\tr\Delta\Gamma=0\),
\begin{equation}
 |\nu(\Phi)-\nu(\Psi)|
 =\big|\tr[\Delta\Gamma\,(\id-\Gamma_\Phi-\Gamma_\Psi)]\big|
 \le\|\Delta\Gamma\|_1
 \le4N\epsilon ,
\end{equation}
using \(0\preceq\Gamma_\Phi+\Gamma_\Psi\preceq2\,\id\), so that
\(\|\id-\Gamma_\Phi-\Gamma_\Psi\|_\infty\le1\).
\end{proof}

Hence, if a protocol of Proposition~\ref{smprop:hybridcost} outputs a
pure fixed-number state \(\Phi\) with
\(\tfrac12\|\ket\Phi\bra\Phi-\ket\Psi\bra\Psi\|_1\le\epsilon\) and \(\nu(\Psi)\ge cN\), then
\(q_0+2\sum_{t\ge1}q_t\ge4(c-4\epsilon)N\), which is linear for every
fixed \(\epsilon<c/4\).

\section{Low-defect inputs and Gaussian-rank compression}
\label{sm:easy}

At low defect, natural-orbital excitation sectors also give useful approximations by short Slater expansions. We first compare product inputs with a Slater reference and then derive a general Gaussian-rank compression bound. The distinction between the existence of a compressed expansion and efficient access to that expansion is essential for sampling conclusions.
For the product family the Slater reference $\Psi_0$ has overlap
$\braket{\Psi_0}{\Psi_{\bm\lambda}}
=\prod_i\sqrt{1-\lambda_i}$, preserved by every passive circuit.
Fuchs--van de Graaf~\cite{Fuchs1999Cryptographic} and data processing
therefore give, for every $V$,
\begin{equation}
 \begin{aligned}
 \dTV(p_{V,\bm\lambda},p_{V,0})
 &\le\sqrt{1-\prod_i(1-\lambda_i)}\\
 &\le\sqrt{\sum_i\lambda_i}
 \le\sqrt{\nu/2}.
 \end{aligned}
 \label{eq:smeasyTV}
\end{equation}
The reference output is determinantal and polynomially
samplable~\cite{Valiant2002Quantum,Terhal2002Classical}. Hence
$\nu=o(1)$ gives an unconditional vanishing-error sampler for an
explicit product input.

For an arbitrary pure input, use the natural-orbital split
$\CC^{2n}=F\oplus G$ of Sec.~\ref{sm:nogo}.  The mean number of
excitations into $G$ is $\mu\le2\nu$ by Eq.~\eqref{eq:smmunu}, so the
weight outside the zero-excitation sector is at most $\mu$.  Thus the
input density operator is within trace distance $\sqrt{2\nu}$ of that of the Slater determinant
filling $F$; the next proposition gives its higher-rank extension.

The fermionic Gaussian rank of a pure state is the minimum number of
normalized pure Gaussian states in a linear expansion, without an
orthogonality requirement~\cite{Cudby2023Gaussian,Dias2024Classical}.

\begin{proposition}[Information-theoretic Gaussian-rank compression]
\label{smprop:compress}
Let $\ket\Psi\in\wedgepow n\CC^{2n}$ be pure and half filled.  For
$0<\epsilon<1$, put
\begin{equation}
 L=\min\!\left\{n,
 \left\lfloor\frac{2\nu}{\epsilon^2}\right\rfloor\right\}.
\end{equation}
There is a normalized state $\Phi$ whose density operator is within trace distance $\epsilon$
of that of $\Psi$ whose fermionic Gaussian rank is at most
\begin{equation}
 \sum_{j=0}^L\binom nj^2
 \le(L+1)n^{2L}
 =n^{O(\nu/\epsilon^2)}.
\end{equation}
Their passive-circuit output distributions are therefore within
$\epsilon$ in total variation.
\end{proposition}

\begin{proof}
With the same natural-orbital split, write
\begin{equation}
 \begin{aligned}
 \Psi&=\sum_{j=0}^n\Psi_j,
 &\Psi_j&\in\wedgepow{n-j}F\otimes\wedgepow jG,\\
 \sum_{j=0}^n j\|\Psi_j\|^2&=\mu\le2\nu,
 &\sum_{j>L}\|\Psi_j\|^2
 &\le\frac{\mu}{L+1}\le\epsilon^2
 \qquad(L<n).
 \end{aligned}
\end{equation}
For $L=n$ the tail vanishes.  The density operator of the normalized projection onto the first
$L+1$ sectors is therefore within trace distance $\epsilon$ of
$\ket\Psi\bra\Psi$, and those sectors have a Slater basis of size
$\sum_{j=0}^L\binom nj^2$.
\end{proof}

Algorithmic use requires the truncated natural-orbital expansion as
input; given it, rank-$R$
Gaussian superpositions are simulable in $\operatorname{poly}(R,n)$
time~\cite{Dias2024Classical}.

This compression is not implied by Gaussian extent $\xi_{\rm FLO}$: for
four-mode products, $\xi_{\rm FLO}=\prod_i(1+\sqrt{m_i})\ge2^\nu$, while the
profile $m_i=1/N$ has $\nu=1$ but
$\log_2\xi_{\rm FLO}=\Theta(\sqrt N)$~\cite{ReardonSmith2024FLOExtent,
Dias2026Optimal}.

\section{Hardness statements and explicit constants}
\label{sm:hardness}

Exact worst-case hardness and approximate-sampling hardness require different ingredients. We first construct a permanent embedding, then specify a finite ensemble with hiding, and finally combine the average-case hardness assumption with anticoncentration to obtain explicit sampling-error tolerances. This order makes the conditional part of the sampling conclusion precise.

\subsection{Worst-case hardness at fixed block parameter}
\label{sm:sharpP}

We embed a permanent into a physical passive-FLO amplitude to establish exact worst-case hardness at every fixed nonzero algebraic block parameter.
For equal blocks write
$\vartheta=\sqrt{\lambda/(1-\lambda)}\in(0,1]$. Let $\mathcal I_x$ be the
$2N$ occupied rows of an outcome $x$. For
$y\in\{0,1\}^N$, let $Y_y$ contain the first orbital pair of block
$i$ when $y_i=1$ and the second pair when $y_i=0$. Then
\begin{equation}
 \matel{x}{\wedgepow{2N}u}{\psi_\lambda^{\otimes N}}
 =(1-\lambda)^{N/2}\mathcal P_x(\vartheta;u),
 \qquad
 \mathcal P_x(\vartheta;u)
 :=\sum_{y\in\{0,1\}^N}\vartheta^{|y|}\det u_{\mathcal I_x,Y_y}.
 \label{eq:smPpoly}
\end{equation}

\begin{lemma}[Orthogonal realization of the selector sum]
\label{smlem:physicaldilation}
Let $\mathbb F$ be either $\mathbb Q$ or a fixed real algebraic number
field. Let $A\in\mathbb F^{4B\times2B}$ have $B$ four-row blocks and
set
\begin{equation}
 \Delta_B(A):=\sum_{y\in\{0,1\}^B}\det A_{Y_y,:}.
 \label{eq:smDB}
\end{equation}
In polynomial time, using a fixed basis of $\mathbb F$, one can
construct $u_A\in O(8B,\mathbb F)$ and a half-filled outcome $x_A$
such that, for $N=2B$,
\begin{equation}
 \mathcal P_{x_A}(1;u_A)
 =2^{2B}\det(I+A^{\mathsf T}A)^{-1}\Delta_B(A).
 \label{eq:smphysicaldilation}
\end{equation}
\end{lemma}

\begin{proof}
Write
\begin{equation}
 G_A=A^{\mathsf T}A,\qquad H_A=(I+G_A)^{-1},\qquad
 \Omega_A=
 \begin{pmatrix}0&-A\\ A^{\mathsf T}&0\end{pmatrix},
 \qquad
 \mathcal C_A=(I-\Omega_A)(I+\Omega_A)^{-1}.
 \label{eq:smcayley}
\end{equation}
Since $\Omega_A$ is skew-symmetric,
$\mathcal C_A\in O(6B,\mathbb F)$. Its last $2B$ columns are
\begin{equation}
 \mathcal V_A=
 \begin{pmatrix}2AH_A\\ (I-G_A)H_A\end{pmatrix},
 \qquad
 \mathcal V_A^{\mathsf T}\mathcal V_A
 =H_A[4G_A+(I-G_A)^2]H_A=I.
 \label{eq:smrationalstiefel}
\end{equation}
Append $I_{2B}$ and reorder the rows of
$\mathcal V_A\oplus I_{2B}$ into $2B$ four-row blocks: with
$\mathcal L_A:=(I-G_A)H_A$ and $(AH_A)_{\langle i\rangle}$ the $i$th four-row
block of $AH_A$, block $i\le B$ and padding block $B+i$ are
\begin{equation}
 \begin{pmatrix}2\,(AH_A)_{\langle i\rangle}&0\end{pmatrix}
 \qquad\text{and}\qquad
 \begin{pmatrix}
  (\mathcal L_A)_{2i-1,:}&0\\
  (\mathcal L_A)_{2i,:}&0\\
  0&e_{2i-1}^{\mathsf T}\\
  0&e_{2i}^{\mathsf T}
 \end{pmatrix},
 \qquad 1\le i\le B.
\end{equation}

After this row permutation, $\mathcal V_A\oplus I_{2B}$ is the last
$4B$ columns of a row-permuted copy
$\widetilde{\mathcal C}_A$ of $\mathcal C_A\oplus I_{2B}$.
Set $u_A=\widetilde{\mathcal C}_A^{\mathsf T}$ and let $x_A$ occupy its last $4B$ rows.
The selected original-block pairs precede the padding identity pairs, so the determinant below has the displayed positive prefactor; no additional row-sign correction is needed. Any selector using the first pair of a
padding block has zero determinant, because the first $B$ selected
pairs already give $2B$ rows supported on the first summand. Thus
only selectors $(y,0^B)$ survive, and
\begin{equation}
 \det (u_A)_{\mathcal I_{x_A},Y_{(y,0^B)}}
 =2^{2B}\det(H_A)\det A_{Y_y,:}.
\end{equation}
Summation over $y$ gives Eq.~\eqref{eq:smphysicaldilation}, and all
matrix operations have polynomial bit complexity over the fixed field
$\mathbb F$.
\end{proof}

\begin{proposition}[Worst-case hardness at fixed $\lambda$]
\label{smprop:sharpP}
Fix $\lambda\in(0,1/2]$ such that
$\vartheta=\sqrt{\lambda/(1-\lambda)}$ is nonzero and algebraic. Exact
evaluation of either
\begin{equation}
 \matel{x}{\wedgepow{2N}u}{\psi_\lambda^{\otimes N}}
 \quad\text{or}\quad
 p_x=\left|\matel{x}{\wedgepow{2N}u}
 {\psi_\lambda^{\otimes N}}\right|^2
\end{equation}
for worst-case $x$ and
$u\in O(4N,\mathbb Q(\vartheta))\subset\U(4N)$, represented in a fixed
basis of $\mathbb Q(\vartheta)$, is $\#\mathrm P$-hard under polynomial-time Turing
reductions.
\end{proposition}

\begin{proof}
Ivanov's reduction maps a $0$--$1$ matrix $M$ to
$A\in\mathbb Z^{4B\times2B}$, with $B$ polynomial in the size of
$M$, such that
$\Delta_B(A)=\operatorname{per}(M)\ge0$~\cite{Ivanov2017Computational,Valiant1979Permanent}.
Multiply one row of the first pair in every block by $\vartheta^{-1}$
and call the resulting matrix $A^{(\vartheta)}$. Then
$\det A^{(\vartheta)}_{Y_y,:}=\vartheta^{-|y|}\det A_{Y_y,:}$.

Apply Lemma~\ref{smlem:physicaldilation} over $\mathbb Q(\vartheta)$. Put the
scaled pairs in the $\ket{1100}$ slots of the first $B$ blocks and the
padding identity pairs in the $\ket{0011}$ slots of the remaining
blocks. The factors $\vartheta^{|y|}$ in Eq.~\eqref{eq:smPpoly} cancel the row
scalings, while every surviving padding choice has weight one. Hence
\begin{equation}
 \mathcal P_{x^{(\vartheta)}}(\vartheta;u^{(\vartheta)})
 =2^{2B}\det\!\left(I+(A^{(\vartheta)})^{\mathsf T}
   A^{(\vartheta)}\right)^{-1}
 \operatorname{per}(M).
 \label{eq:smhardembedding}
\end{equation}
Here $N=2B$ and
$(1-\lambda)^{N/2}=(1+\vartheta^2)^{-B}\in\mathbb Q(\vartheta)$. The prefactor is
positive and exactly computable, so the amplitude recovers
$\operatorname{per}(M)$. From the probability one recovers its square
and then its nonnegative integer square root. The field degree is
fixed, and all exact data have polynomial bit length.
\end{proof}

For algebraic $\gamma=4\lambda(1-\lambda)\in(0,1]$, one has
$\vartheta=(1-\sqrt{1-\gamma})/\sqrt\gamma$, which is nonzero and algebraic.
Thus Proposition~\ref{smprop:sharpP} applies at every fixed algebraic defect
density.

At maximal density, approximate-sampling hardness additionally uses
the average-case assumption below.

\subsection{A finite ensemble with hiding}
\label{sm:scales}

We discretize the Haar ensemble while preserving the probability moments and the permutation symmetry needed for hiding.

\smallskip
\paragraph*{Finite-precision Haar coupling.}
Use the Haar parametrization of $\U(4N)$ in the corrected
Appendix B of Ref.~\cite{Pozniak1998Composed}
(\href{https://arxiv.org/abs/chao-dyn/9707006v2}{arXiv v2}), with $O(N^2)$ independent
parameters. For each parameter one samples one of $2^{b_N}$
equal-mass bins, $b_N=\lceil cN^2\rceil$. Fix a deterministic
implementation of dyadic bisection in the stereographic parameter
toward the bin's probability midpoint, stopping when a certified CDF
interval of width at most $2^{-b_N-4}$ lies wholly in the middle half
of the bin's probability interval. For mixing angles use $s\in[0,1]$;
for phases use $s\in[-1,1]$ in the two charts centered at $\pi/2$ and
$3\pi/2$, selecting the chart by the bin's half-circle.
For a mixing angle,
\begin{equation}
 (\cos\theta,\sin\theta)
 =\left(\frac{1-s^2}{1+s^2},\frac{2s}{1+s^2}\right),
 \qquad s\in\mathbb Q.
\end{equation}
The mixing variable $w=\sin^2\theta$ has a $\mathrm{Beta}(a,b)$
cumulative distribution with integer parameters $a,b=O(N)$,
rational coefficients, and degree
$a+b-1=O(N)$. The density is at most
$B(a,b)^{-1}\le2^{O(N)}$, so a bin of mass $2^{-b_N}$ has width at
least $2^{-b_N-O(N)}$ and $O(b_N)$ bisection steps suffice. Uniform
phase bins have width $O(2^{-b_N})$. The resulting unitary $\widehat u$
uses $O(N^4)$ random bits and has rational real and imaginary parts of
polynomial bit length.

For an equal-mass beta bin $[w_-,w_+]$, set $\Delta w=w_+-w_-$ and
$d=a+b-1$. On the middle half of this interval in $w$, both $w$ and $1-w$ are at least $\Delta w/4$,
so, using $B(a,b)\le1$,
\begin{equation}
 2^{-b_N}\ge\frac{\Delta w}{2B(a,b)}\left(\frac{\Delta w}4\right)^{d-1}
 \ge\frac{\Delta w}2\left(\frac{\Delta w}4\right)^{d-1}.
\end{equation}
Together with
$|\arcsin\sqrt{w_+}-\arcsin\sqrt{w_-}|\le2\sqrt{\Delta w}$, this implies
\begin{equation}
 \begin{aligned}
 \Delta w
 &\le8\bigl(2^{1-b_N}\bigr)^{1/(a+b-1)},\\
 |\arcsin\sqrt{w_+}-\arcsin\sqrt{w_-}|
 &\le4\sqrt2\,
 \bigl(2^{1-b_N}\bigr)^{1/[2(a+b-1)]}
 =2^{-\Omega(N)}.
 \end{aligned}
\end{equation}
Phase bins are smaller.  Couple $\widehat u$ to a Haar unitary $u$ by
drawing each continuous parameter from its sampled bin.  Telescoping
over the $O(N^2)$ Givens factors and taking $c$ sufficiently large
gives
\begin{equation}
 \|\widehat u-u\|_{\rm op}
 =O\!\left(N^2 2^{-\Omega(N)}\right)
 \le D_N^{-2}.
 \label{eq:smquantcoupling}
\end{equation}

\smallskip
\paragraph*{Permutation hiding.}
Draw $\lceil\log_2((4N)!D_N^3)\rceil$ fair bits, reduce the resulting
integer modulo $(4N)!$, and convert that rank to a permutation
$\pi\in\mathfrak S_{4N}$ by Lehmer unranking.  The residue classes
differ in size by at most one, so
\begin{equation}
 d_{\rm TV}(\pi,\mathrm{Unif}(\mathfrak S_{4N}))
 \le\tau_N:=
 (4N)!\,2^{-\lceil\log_2((4N)!D_N^3)\rceil}
 \le D_N^{-3}.
\end{equation}
Set $\widetilde u=\mathsf P_\pi\widehat u$ and let $\mathcal D_N$ be
its law. Under a maximal coupling with a uniform permutation
$\pi_{\rm H}$, the unitary
$u_{\rm H}=\mathsf P_{\pi_{\rm H}}u$ is Haar and, except with
probability $\tau_N$,
\begin{equation}
 \left\|\wedgepow{2N}\widetilde u
       -\wedgepow{2N}u_{\rm H}\right\|_{\rm op}
 \le2ND_N^{-2},\qquad
 |p_x(\widetilde u)-p_x(u_{\rm H})|\le4ND_N^{-2}.
 \label{eq:smwedgecoupling}
\end{equation}
The ensemble $\mathcal D_N$ is efficiently samplable and left
mode-permutation invariant up to $O(\tau_N)$ in total variation.

\begin{lemma}[Moment transfer and hiding]
\label{smlem:fphiding}
For $X_x(v)=D_Np_x(v)$ and every fixed $x$,
\begin{equation}
 \E_{v\sim\mathcal D_N}X_x(v)=1+o(1),\qquad
 \E_{v\sim\mathcal D_N}X_x(v)^2=\CN_N+o(1).
 \label{eq:smfpmoments}
\end{equation}
If a classical sampler with output law $q_v$ satisfies
$\E_{v\sim\mathcal D_N}\|q_v-p_v\|_1\le\varepsilon$, then for every
fixed half-filled outcome $x_0$ there is an efficient randomized
experiment, with acceptance probability $\bar q_{x_0}(v)$, such that
\begin{equation}
 \E_{v\sim\mathcal D_N}
 |\bar q_{x_0}(v)-p_v(x_0)|
 \le\frac{\varepsilon+o(1)}{D_N}.
 \label{eq:smfphiding}
\end{equation}
\end{lemma}

\begin{proof}
Write $X=X_x(\widetilde u)$, $Y=X_x(u_{\rm H})$, and
$\delta_N=4N/D_N$, and let $\mathcal G$ be the nonexceptional event in
Eq.~\eqref{eq:smwedgecoupling}. On $\mathcal G$, $|X-Y|\le\delta_N$ and
$X\le Y+\delta_N$; since $\E Y=1$,
$\E[\mathbf 1_{\mathcal G}(X+Y)]\le2+\delta_N$. On $\mathcal G^c$, both variables are at
most $D_N$. Therefore
\begin{equation}
 \begin{aligned}
 |\E X-1|
 &\le\delta_N+\tau_ND_N=o(1),\\
 |\E X^2-\E Y^2|
 &\le\delta_N(2+\delta_N)
 +2\tau_ND_N^2=o(1).
 \end{aligned}
\end{equation}

For hiding, draw an independent approximately uniform permutation
$\sigma$, run the sampler on $\mathsf P_\sigma v$, and accept iff its
output is $\mathsf P_\sigma x_0$. Thus
\begin{equation}
 \bar q_{x_0}(v)
 =\E_\sigma q_{\mathsf P_\sigma v}(\mathsf P_\sigma x_0).
\end{equation}
Mode covariance gives
$p_{\mathsf P_\sigma v}(\mathsf P_\sigma x_0)=p_v(x_0)$. Jensen's
inequality, approximate left invariance, and uniformity of the orbit
of $x_0$ give
\begin{equation}
 \E_v|\bar q_{x_0}(v)-p_v(x_0)|
 \le\frac1{D_N}\E_{v'}\|q_{v'}-p_{v'}\|_1+O(\tau_N),
\end{equation}
which proves Eq.~\eqref{eq:smfphiding}. The experiment uses
polynomially many random bits and is therefore admissible for
Stockmeyer counting.
\end{proof}

The additive-error and coarse-marginal guarantees of
Refs.~\cite{Hakkaku2022Quantifying,ReardonSmith2024Improved,
Dias2024Classical,Oh2026Classical} do not by themselves provide the
relative-error guarantee assumed below.

\subsection{Probability estimation and sampling hardness}
\label{sm:optconst}

The following average-case conjecture supplies the additional ingredient needed to turn probability estimation into a sampling-hardness statement.

\begin{conjecture}[Average-case probability hardness at maximal density]
\label{smconj:achard}
Fix the input $\ket{\psi_{1/2}}^{\otimes N}$ and a half-filled
outcome $x$. Suppose a randomized procedure, given
$\widetilde u\sim\mathcal D_N$, estimates $p_x(\widetilde u)$ to
relative error $1/4$ with joint probability at least $2/3$ over the
instance and its internal randomness, for all sufficiently large $N$.
Then a polynomial-time randomized Turing reduction using this
procedure as an oracle computes a $\#\mathrm P$-complete function on
every input with success probability at least $2/3$.
\end{conjecture}

The success fraction is essential. At maximal density the constant
estimate $1/D_N$ already succeeds to relative error $1/4$ with limiting
probability $e^{-4/5}-e^{-4/3}\simeq0.186$. We use the high-success
assumption above separately from the worst-case result.

\begin{corollary}[Conditional sampling hardness at maximal density]
\label{cor:hardness}
Assume Conjecture~\ref{smconj:achard}. For the input
$\ket{\psi_{1/2}}^{\otimes N}$, no polynomial-time classical sampler
can satisfy
$\E_{\widetilde u\sim\mathcal D_N}
\|q_{\widetilde u}-p_{\widetilde u}\|_1\le1/64$
for all sufficiently large $N$, unless the polynomial hierarchy
collapses to its third level.
\end{corollary}
\begin{proof}
For any pure input and a sampler with ensemble-averaged $\ell_1$
error at most $\varepsilon$, write
$X=D_Np_{x_0}(v)$ and $Y=D_N\bar q_{x_0}(v)$ for $v\sim\mathcal D_N$.
Lemma~\ref{smlem:fphiding} gives
$\E|Y-X|\le\varepsilon+o(1)$ and $\E Y=O(1)$.
Stockmeyer counting~\cite{Stockmeyer1985On} with relative error
$N^{-2}$ and failure probability $D_N^{-2}$, with its probability
estimate clipped to $[0,1]$ and rescaled by $D_N$, produces
$\widehat X$ satisfying
\begin{equation}
 \E|\widehat X-Y|
 \le N^{-2}\E Y+D_N^{-1}=o(1),\qquad
 \E|\widehat X-X|\le\varepsilon+o(1).
 \label{eq:smstockrisk}
\end{equation}
Amplification takes $O(\log D_N)=O(N)$ repetitions, so this is an
$\mathrm{FBPP}^{\mathrm{NP}}$ procedure. Expectations include its
internal randomness.

For $\mathcal B=\{|\widehat X-X|>X/4\}$ and any fixed $t>0$, the
pointwise bound $t\mathbf1_{\mathcal B}\le4|\widehat X-X|+(t-X)_+$, where
$(z)_+:=\max\{z,0\}$, gives
\begin{equation}
 \Pr(\mathcal B^c)\ge
 \frac{\E\min(X,t)-4\varepsilon}{t}-o(1).
 \label{eq:smrisktrunc}
\end{equation}
At $\gamma=1$, Theorem~\ref{thm:PT}, Sec.~\ref{sm:fracpt}, and
Eq.~\eqref{eq:smwedgecoupling} give
$X\Rightarrow\operatorname{Exp}(1)$ on $\mathcal D_N$.
Taking $\varepsilon=1/64$ and $t=\ln(3/2)$ yields
$\E\min(X,t)=1-e^{-t}+o(1)=1/3+o(1)$, and hence
\begin{equation}
 \Pr(\mathcal B^c)\ge\frac{13}{48\ln(3/2)}-o(1)
 =0.667957\ldots-o(1)>\frac23.
 \label{eq:smptsuccess}
\end{equation}
Conjecture~\ref{smconj:achard} puts exact $\#\mathrm P$ computation
in $\mathrm{FBPP}^{\mathrm{NP}}$. Amplifying this worst-case counting
procedure and applying Toda's theorem give
$\mathrm{PH}\subseteq\mathrm{BPP}^{\mathrm{NP}}
\subseteq\Sigma_3^{\mathrm P}$.
\end{proof}

\begin{proposition}[Density-dependent sampling tolerance]
\label{smprop:densitytolerance}
For $0<\gamma\le1$, define
\begin{equation}
 \varepsilon_*(\gamma):=\frac14
 \left[\inf_Z\E\min(Z,t)-\frac{2t}{3}\right]_+,
 \qquad t=\ln(3/2),
 \label{eq:smdensitytolerance}
\end{equation}
where the infimum is over nonnegative random variables with
$\E Z^\ell=\ell!\gamma^{-\binom{\ell}{2}}$ for every integer $\ell\ge1$.
This certified tolerance is nondecreasing in $\gamma$, strictly increasing
on its positive region, and
\begin{equation}
 \lim_{\gamma\uparrow1}\varepsilon_*(\gamma)
 =\varepsilon_*(1)
 =\frac{1-2\ln(3/2)}{12}>\frac1{64}.
 \label{eq:smtoleranceendpoint}
\end{equation}
In particular, $\varepsilon_*(\gamma)>7\times10^{-4}$ for
$9/10\le\gamma\le1$.
For a fixed algebraic $\gamma$, assume the analogue of
Conjecture~\ref{smconj:achard} for
$\ket{\psi_\lambda}^{\otimes N}$ with
$\lambda=(1-\sqrt{1-\gamma})/2$, retaining relative error $1/4$
and joint success probability $2/3$.
For any fixed $0\le\varepsilon<\varepsilon_*(\gamma)$, no polynomial-time
classical sampler can achieve ensemble-averaged $\ell_1$ error at most
$\varepsilon$ for all sufficiently large $N$,
unless the polynomial hierarchy collapses to its third level.
\end{proposition}

\begin{proof}
The admissible laws are nonempty by Eq.~\eqref{eq:smlognormal}.
They are weakly compact. Their mean ensures tightness, and each
fixed higher moment makes the preceding moment uniformly integrable,
so all moment constraints pass to weak limits.
Bounded continuity of $\min(\cdot,t)$ ensures that the infimum is attained.

For $\gamma_1<\gamma_2$, take a minimizer $Z$ at $\gamma_2$ and an
independent $L_{\gamma_1/\gamma_2}$ from Eq.~\eqref{eq:smlognormal}.
Their product has the required moments at $\gamma_1$.
Conditional Jensen gives
$\E\min(ZL_{\gamma_1/\gamma_2},t)<\E\min(Z,t)$.
Strictness follows from $\Pr(Z>0)>0$ and the full positive support
of the mean-one lognormal factor. This orders the infima and hence
their positive parts. As $\gamma\uparrow1$, every subsequential
limit of minimizing laws has moments $\ell!$ and is exponential
by Carleman's criterion. Boundedness of $\min(Z,t)$ gives
Eq.~\eqref{eq:smtoleranceendpoint}, including positivity near one.

For an explicit bound, put $\varpi(x)=1-9x/8+x^2/4-x^3/80$.
On $[0,t]$, $\varpi(x)\ge1-9x/8\ge0$ and $9t/4<1$, so
$1-\varpi(x)^2\le9x/4\le x/t$.
For $x\ge t$, the inequality $1-\varpi(x)^2\le1$ suffices. Hence
\begin{equation}
 \min(x,t)\ge t[1-\varpi(x)^2]\qquad(x\ge0).
 \label{eq:smtruncminorant}
\end{equation}
At $\gamma=9/10$, the first six prescribed moments give
$\E\varpi(Z)^2<3263/10^4$.
Since $t>2/5$, Eq.~\eqref{eq:smdensitytolerance} yields
$\varepsilon_*(9/10)>7\times10^{-4}$.
Monotonicity extends this bound to $\gamma\ge9/10$.

By Theorem~\ref{thm:PT} and the same higher-moment argument,
every subsequential Haar limit at fixed density has the prescribed
moments. The coupling in Eq.~\eqref{eq:smwedgecoupling} gives the
same weak limits on $\mathcal D_N$. Thus, whenever
$\varepsilon_*(\gamma)>0$,
$\liminf_N\E\min(X,t)\ge2t/3+4\varepsilon_*(\gamma)$.
For $\varepsilon<\varepsilon_*(\gamma)$,
Eq.~\eqref{eq:smrisktrunc} gives success probability strictly above
$2/3$ for all large $N$.
The counting reduction of Corollary~\ref{cor:hardness} completes the proof.
\end{proof}

\subsection{A per-circuit count of heavy outputs}
\label{sm:typical}

The ensemble collision moment also bounds the fraction of heavy outcomes for most individual circuits.
For a fixed circuit define
\begin{equation}
 \operatorname{Col}(V)
 :=D_N\sum_xp_V(x)^2
 =\E_{x\sim{\rm unif}}X_x(V)^2.
\end{equation}
Since $\E_V\operatorname{Col}(V)=\CN_N$, at least a $1-\delta$
fraction of circuits obey
$\operatorname{Col}(V)\le\CN_N/\delta$. For each such circuit,
Paley--Zygmund over a uniform outcome gives
\begin{equation}
 \frac1{D_N}\big|\{x:X_x(V)\ge\alpha\}\big|
 \ge\frac{\delta(1-\alpha)^2}{\CN_N},
 \qquad 0<\alpha,\delta<1.
 \label{eq:smtypical}
\end{equation}
For bounded-block products with
$\nu/N\to\bar\gamma>0$, the right-hand side is
$\delta\bar\gamma(1-\alpha)^2/2+o(1)$.

\section{Dephasing and the noisy collision law}
\label{sm:dephasing}

The pure-input assumption matters when interpreting the defect as a non-Gaussian resource. Blockwise dephasing preserves the one-body occupations while changing the replica-singlet weight and suppressing the exchange contribution to the collision law. We derive this change explicitly and show why the defect and collision moment alone are insufficient certificates for mixed inputs.
For any mixed state $\rho$ supported on the half-filling sector,
\begin{equation}
 \E_Vp_x^2=
 \frac{\tr(\Pi_2\rho^{\otimes2})}
 {\dim\Schur_{(2^{2N})}\CC^{4N}},
\end{equation}
and zero replica charge reduces $\Pi_2$ to the same polar average as
in the pure case.  We apply blockwise dephasing in the canonical
Slater-pair basis before the passive circuit, giving
\begin{align}
 &\rho_i=(1-\lambda_i)\ket{0011}\!\bra{0011}
 +\lambda_i\ket{1100}\!\bra{1100}+(1-\eta_i)\sqrt{\lambda_i(1-\lambda_i)}
 \big(\ket{0011}\!\bra{1100}+\mathrm{h.c.}\big),
 \qquad 0\le \eta_i\le1,\nonumber\\
 &\rho=\bigotimes_i\rho_i,
  \label{eq:smdephasedblock}
\end{align}
where $\eta_i=0$ is the coherent block and $\eta_i=1$ removes its
off-diagonal coherence.  The two experimental copies are prepared
independently in the same state $\rho$.

With $z=\cos\varphi$, the polar average gives
\begin{equation}
 \CN_N=\kappa_N\,\frac12\int_{-1}^1
 \prod_{i=1}^N
 \left\{1-\frac{m_i}{8}(1-z)
 [3+z-(1-\eta_i)^2(1-z)]\right\}\,dz .
 \label{eq:smnoisy}
\end{equation}
For $\eta_i=0$ this reduces to Eq.~\eqref{eq:integral}.

The $z=1$ identity saddle has noise-independent decay rate $\nu/2$ and
contributes $\nu^{-1}[1+o(1)]$ to the polar integral.  At the $z=-1$
exchange saddle,
\begin{equation}
 \begin{aligned}
 P_-&:=\prod_i\left[1-\frac{m_i}{2}\eta_i(2-\eta_i)\right]
      =\tr\rho^2,\\
 S_-&:=\sum_i\frac{m_i(1-\eta_i)^2}
      {2-m_i\eta_i(2-\eta_i)} .
 \end{aligned}
 \label{eq:smexchangeparams}
\end{equation}
If $\nu\to\infty$ and $\eta_i\le\eta_0<1$ on every nonfree block,
then $S_-=\Theta(\nu)$. Each block factor in Eq.~\eqref{eq:smnoisy}
is quadratic in $z$ and bounded below by $1/2$, so
\begin{equation}
 \CN_N
 =N\left(\frac1\nu+\frac{P_-}{2S_-}\right)[1+o(1)] .
 \label{eq:smnoisyhardedge}
\end{equation}
If $\eta_i=\eta$ for all $i$ and $\eta\nu\to c\in(0,\infty)$,
$P_-\to e^{-c}$ and $2S_-=\nu+O(1)$, giving
\begin{equation}
 \CN_N=(1+e^{-c})\frac N\nu[1+o(1)] .
 \label{eq:smdephasingcrossover}
\end{equation}
At any fixed $\eta\in(0,1)$ and fixed positive $\nu/N$, however,
$P_-=e^{-\Theta(N)}$. The exchange saddle disappears, and the pure-state
leading coefficient is reduced from $2N/\nu$ to $N/\nu$.

At $\eta=1$ the state is convex-Gaussian: a classically samplable
mixture of product Slater determinants. A one-endpoint analysis of
Eq.~\eqref{eq:smnoisy} gives
\begin{equation}
 \CN_N=\frac N\nu[1+o(1)]\longrightarrow\frac1{\bar\gamma}
 \quad\text{if}\quad \nu/N\to\bar\gamma>0 .
 \label{eq:smfulldephasing}
\end{equation}
In particular $\bar\gamma=1/2$ reproduces the Porter--Thomas collision
value $2$. Thus for mixed states neither the covariance defect nor the
collision moment alone certifies non-convex-Gaussianity or higher-order
Porter--Thomas statistics.

\section{Numerical protocols}
\label{sm:numerics}

The numerical checks use different protocols for collision moments, higher integer moments, and fractional moments. We specify the analytic integrals or sampled amplitudes used in each case, together with the sample counts and uncertainty estimates. These details distinguish numerical evidence for the Mellin conjecture from the rigorous moment results.
All collision curves use the one-dimensional integral
\eqref{eq:integral}; for uniform profiles we use the terminating
hypergeometric form \eqref{eq:sm2F1}.

For the third moment we use
$\E_VX_x^3=R_3(N)\A_3$, with $R_3$ from Eq.~\eqref{eq:smRq}, and
sample the flag integral \eqref{eq:smflag} using Haar-random $\U(3)$ matrices.
Each block symbol reduces to
\begin{equation}
 \phi^{(3)}_{\lambda_i}(g)
 =1-\frac{m_i}{4}
 \left(3-\sum_{a,b=1}^3|g_{a,b}|^4\right).
\end{equation}
The theory curves use $4\times10^6$ flag samples, with relative
standard error below $1.2\%$ at the largest displayed $N$ and smaller
elsewhere.  The direct-Haar markers in
Fig.~\ref{fig:cost}(c) use $2\times10^5$ circuits, with the displayed
$95\%$ batch intervals.

Figure~\ref{smfig:mellin} uses direct floating-point amplitude
evaluations. Each draw uses a reduced Haar-random $2N$-frame in
$\CC^{4N}$, and the determinant sweep costs $O(2^N N^3)$. We evaluate
the $2^N$ weighted $2N\times2N$ minors in vectorized complex-double
batches and combine the batch sums with compensated summation. At fixed
$N$, the same sweep supplies all displayed values of $\nu$ by changing
only the block weights. As checks, the samples reproduce
$\E_VX_x=1$ and the closed-form collision moment.

For $N=12,14,16,18,20,22$, respectively, the numbers of independent
Haar draws are $(2.0,1.6,1.2,0.8,0.6,0.4)\times10^5$, for a total of
$6.6\times10^5$, corresponding to
$2.629632\times10^{11}$ determinant evaluations. Panels (a) and (b) use
two-sided $95\%$ Student-$t$
intervals from $20$ equal batches.  In panel (c), for each $N$ and
$r=1/4,1/2,3/4$, we fit the least-squares slope of
$\log\widehat{\E}_VX_x^r$ against $\log\gamma$ over
$4\le\nu\le0.6N$.  Its pointwise $95\%$ interval comes from $1000$
joint resamples, across $\nu$, of the same $20$ batches.  The
displayed bar is the interval width divided by $2\times1.96$, a
Gaussian-equivalent visual scale rather than a bootstrap standard
deviation.  On the sampled even-integer defect grid, the window at $N=12$ consists only of $\nu=4,6$, so the
reported slope is their secant.
\end{document}